\RequirePackage{silence}
\documentclass[a4paper, UKenglish, cleveref, autoref,% numberwithinsect,
thm-restate]{lipics-v2019}
\usepackage{mathpartir}
\usepackage{thmtools}
\usepackage{mathtools}
\usepackage{multicol}
\usepackage[bbgreekl]{mathbbol}
\usepackage{bbm}
\usepackage{amsmath}
\usepackage{amssymb}
\usepackage{tikz}
\usetikzlibrary{arrows,fit,calc}
\usepackage{array}
\usepackage{scalerel}
\usepackage{float}
\usepackage{thm-restate}
\usepackage{verbdef}

\DeclareSymbolFont{stmry}{U}{stmry}{m}{n}
\DeclareMathDelimiter\llbracket{\mathopen}{stmry}{"4A}{stmry}{"71}
\DeclareMathDelimiter\rrbracket{\mathclose}{stmry}{"4B}{stmry}{"79}
\DeclareMathSymbol\llparenthesis\mathopen{stmry}{"4C}
\DeclareMathSymbol\rrparenthesis\mathclose{stmry}{"4D}
\DeclareMathSymbol\subsetplus\mathrel{stmry}{"44}
\DeclareMathSymbol\supsetplus\mathrel{stmry}{"45}

\DeclareMathAlphabet{\mathpzc}{OT1}{pzc}{m}{it}

\newcommand\Mid{\mathrel{|}}
\newcommand\Coloneqq{\mathrel{\mathop{::}}=}
\newcommand\paren[1]{\left(#1\right)}
\renewcommand\brack[1]{\left[#1\right]}
\newcommand\set[1]{\left\{#1\right\}}
\newcommand\setcompr[2]{\left\{#1~\middle|~#2\right\}}
\newcommand\tuple[1]{\left\langle#1\right\rangle}
\newcommand\eqdef{\mathrel{\mathop{:}}=}
\newcommand\pset[1]{\mathcal{P}\paren{#1}}
\newcommand\fpset[1]{\mathcal{P}_f\paren{#1}}
\newcommand\Nat{\mathbb{N}}
\newcommand\rsem[1]{\stretchleftright{\llparenthesis}{\displaystyle \makebox[0pt][c]{\color{white} $\beta$}{#1}}{\rrparenthesis}}

\renewcommand\epsilon\varepsilon

\newcommand\psem[2][\sigma]{\left\llbracket {#2}\right\rrbracket_{#1}}
\newcommand\sem[1]{\left\llbracket {#1}\right\rrbracket}

\newcommand\Id[1]{Id_{#1}}

\newcommand\alphabet\Sigma

\newcommand\events[1][]{\mathcal{E}_{#1}}
\newcommand\order[1][]{\mathrel{\leq_{#1}}}
\newcommand\labels[1][]{\lambda_{#1}}
\newcommand\boxes[1][]{\mathcal{B}_{#1}}

\newcommand\posets[1][\alphabet]{\mathbb{P}_{#1}}
\newcommand\pomsets[1][\alphabet]{\text{\textnormal{\textbf{Pom}}}_{#1}}

\newcommand\pomeq{\cong}
\newcommand\subsetsim{%
  \mathrel{%
    \ooalign{%
      \raise0.2ex\hbox{$\subset$}\cr\hidewidth\raise-0.8ex\hbox{\scalebox{0.9}{$\sim$}}%
      \hidewidth\cr%
    }%
  }%
}
\newcommand\revsubsume{\sqsupseteq}
\newcommand\subsume{\sqsubseteq}
\newcommand\osubsume{\sqsubseteq_o}
\newcommand\bsubsume{\sqsubseteq_b}

\newcommand\clpom[1]{{#1{\downarrow}}}
\newcommand\upclpom[1]{{#1{\uparrow}}}

\newcommand\homo{\to}

\newcommand\atom[1]{\mathbb{#1}}
\newcommand\unitposet{\mathbb{\bbespilon}}
\newcommand\pomseq\otimes
\newcommand\pompar\oplus
\newcommand\boxing[1]{\brack{#1}}

\newcommand\issubpom{\mathrel{\subsetplus}}
\newcommand\contains{\mathrel{\supsetplus}}
\newcommand\restrict[2]{{#1}{\downharpoonright}_{#2}}
\newcommand\spterms[1][\alphabet]{\mathtt{SP}_{#1}}
\newcommand\terms[1][\alphabet]{\mathtt{T}_{#1}}

\newcommand\join{\mathop{+}}
\newcommand\tpar{\mathop{\parallel}}
\newcommand\tseq{\mathop{;}}

\newcommand\bigjoin{\sum}

\newcommand\proj[2]{{\pi_{#1}\paren{#2}}}

\newcommand\Kformulas[1][\alphabet]{\mathtt{F}_{#1}}
\newcommand\Jformulas[1][\alphabet]{\mathtt{F}_{#1}^+}

\newcommand\K{{\pomeq}}
\newcommand\Ju{{\revsubsume}}
\newcommand\Jd{{\subsume}}

\newcommand\satK{\mathrel{\models_\K}}
\newcommand\satJu{\mathrel{\models_{\Ju}}}
\newcommand\satJd{\mathrel{\models_{\Jd}}}
\newcommand\satKU{\mathrel{\models_\K^\forall}}
\newcommand\satJuU{\mathrel{\models_{\Ju}^\forall}}
\newcommand\satJdU{\mathrel{\models_{\Jd}^\forall}}
\newcommand\satKE{\mathrel{\models_\K^\exists}}
\newcommand\satJuE{\mathrel{\models_{\Ju}^\exists}}
\newcommand\satJdE{\mathrel{\models_{\Jd}^\exists}}

\newcommand\then{\mathop{\blacktriangleright}}
\newcommand\nextto{\mathop{\star}}
\newcommand\context[1]{\rsem{#1}}

\newcommand\emptyprop\bot

\newcommand\independent[1][\pomeq]{\mathrel{\#_{#1}}}

\newcommand\axeq{\mathrm{BiMon}_{\Box}}
\newcommand\axinf{\mathrm{CMon}_{\Box}}
\newcommand\axsr{\mathrm{SR}_{\Box}}
\newcommand\axsrinf{\mathrm{CSR}_{\Box}}

\newcommand\Read[1]{\text{\includegraphics[width=3ex]{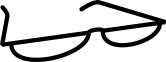}}_{#1}}
\usepackage{marvosym}
\newcommand\Print{\text{\Printer}}
\newcommand\Write[1]{\text{\WritingHand}_{#1}}
\newcommand\Compute[1]{\text{\Industry}_{#1}}

\newcommand\VoteProc{\mathtt{VoteProc}}
\newcommand\Choose{\mathtt{Choose}}
\newcommand\Publish{\mathtt{Publish}}
\renewcommand\choose[1]{\text{\Letter}_{#1}}
\newcommand\vote[1]{\mathtt{Vote}\paren{#1}}
\newcommand\send[1]{\text{\Telefon}_{#1}}

\newcommand\conflict[1]{\mathbf{conflict}_{{#1}}}
\newcommand\seqsep{\mathbf{SendAfterVote}}
\newcommand\fchoose[1]{\mathbf{choose}_{{#1}}}
\newcommand\votethensend{\mathbf{VoteThenSend}}

\newcounter{mysubtable}
\newcounter{myaxiom}
\newcounter{myeq}
\newcommand\axlabel[1]{%
  \stepcounter{myaxiom}%
  \label{#1}%
  \tag{\Alph{mysubtable}{\scriptsize{\arabic{myaxiom}}}}%
}
\newenvironment{tableequations}{%
  \setcounter{myaxiom}{0}%
  \stepcounter{mysubtable}%
  \ignorespaces
}{% 
  \ignorespacesafterend
}

\numberwithin{equation}{section}

\theoremstyle{plain}
\newtheorem{fact}[theorem]{Fact}

\title{Pomsets with Boxes: Protection, Separation, \\ and Locality in Concurrent Kleene Algebra}

\titlerunning{Pomsets with Boxes}

\author{Paul Brunet}{University College London, UK \and \url{paul.brunet-zamansky.fr}}{paul@brunet-zamansky.fr}{https://orcid.org/0000-0002-9762-6872}{}%TODO mandatory, please use full name; only 1 author per \author macro; first two parameters are mandatory, other parameters can be empty. Please provide at least the name of the affiliation and the country. The full address is optional

\author{David Pym}{University College London, UK \and \url{www.cantab.net/users/david.pym/} }{d.pym@ucl.ac.uk}{https://orcid.org/0000-0002-6504-5838}{}

\authorrunning{P. Brunet and D. Pym} %TODO mandatory. First: Use abbreviated first/middle names. Second (only in severe cases): Use first author plus 'et al.'

\Copyright{Paul Brunet and David Pym} %TODO mandatory, please use full first names. LIPIcs license is "CC-BY";  http://creativecommons.org/licenses/by/3.0/

\ccsdesc[500]{Theory of computation~Logic and verification}
\ccsdesc[500]{Theory of computation~Semantics and reasoning}
\ccsdesc[500]{Theory of computation~Separation logic}

\keywords{Concurrent Kleene Algebra,
  Pomsets,
  Atomicity,
  Semantics,
  Separation,
  Local reasoning,
  Bunched logic, 
  Frame rules.}

\category{} %optional, e.g. invited paper

\relatedversion{This preprint is based on a paper to appear in FSCD 2020.}

\supplement{ A formalization of Section~\ref{sec:algebra} in Coq is
  available on github:
  \href{https://github.com/monstrencage/AtomicCKA}{\texttt{AtomicCKA}}.}%optional, e.g. related research data, source code, ... hosted on a repository like zenodo, figshare, GitHub, ...

\funding{This work has been supported by UK EPSRC Research Grant EP/R006865/1:
  Interface Reasoning for Interacting Systems (IRIS).}%optional, to capture a funding statement, which applies to all authors. Please enter author specific funding statements as fifth argument of the \author macro.

\acknowledgements{The authors are grateful to their colleagues at UCL
  and within IRIS project for their interest. We also thank the
  anonymous referees for their comments and suggestions.}%optional

\nolinenumbers %uncomment to disable line numbering
\hideLIPIcs

\EventEditors{Zena M. Ariola}
\EventNoEds{1}
\EventLongTitle{5th International Conference on Formal Structures for Computation and Deduction (FSCD 2020)}
\EventShortTitle{FSCD 2020}
\EventAcronym{FSCD}
\EventYear{2020}
\EventDate{June 29--July 5, 2020}
\EventLocation{Paris, France}
\EventLogo{}
\SeriesVolume{167}
\ArticleNo{5}
\begin{document}

\maketitle

% TODO mandatory: add short abstract of the document
\begin{abstract}
  Concurrent Kleene Algebra is an elegant tool for equational
  reasoning about concurrent programs.
  An important feature of concurrent programs that is missing from CKA
  is the ability to restrict legal interleavings. To remedy this we
  extend the standard model of CKA, namely pomsets, with a new
  feature, called boxes, which can specify that part of the system is
  protected from outside interference. We study the algebraic properties
  of this new model.
  Another drawback of CKA is that the language used for expressing
  properties of programs is the same as that which is used to express
  programs themselves.
  This is often too restrictive for practical purposes.
  We provide a logic, `pomset logic', that is an assertion language
  for specifying such properties, and which is interpreted on pomsets
  with boxes.
  In contrast with other approaches, this logic is not state-based,
  but rather characterizes the runtime behaviour of a program.
  We develop the basic metatheory for the relationship between pomset
  logic and CKA, including frame rules to support local reasoning,  and 
  illustrate this relationship with simple examples.

\end{abstract}

%%%% TODO : many references

\section{Introduction} \label{sec:intro}

Concurrent Kleene Algebra (CKA)
\cite{hoare-moeller-struth-wehrman-2009,kbswz20,kbsz18,bps17} is an
elegant tool for equational reasoning about concurrent programs. Its
semantics is given in terms of pomsets languages; that is, sets of
pomsets. Pomsets~\cite{gischer88}, also known as partial
words~\cite{grabowski-1981}, are a well-known model of concurrent
behaviour, traditionally associated with runs in Petri
nets~\cite{jategaonkarDecidingTrueConcurrency1996,bps17}.

However, in CKA the language used for expressing properties of
programs is the same as that which is used to express programs
themselves.
It is clear that this situation is not ideal for specifying and reasoning about 
properties of programs. Any language specifiable in CKA terms has bounded width 
(i.e., the number of processes in parallel; the size of a maximal independent set)  
% independent=no ordereds pairs
and bounded depth (i.e., the number of alternations of parallel and
sequential compositions)\cite{ls14}. However, many properties of
interest --- for example, safety properties --- are satisfied by sets of
pomsets with both unbounded width and depth.

In this paper, we provide a logic, `pomset logic', that is an assertion language for 
specifying such properties. We develop the basic metatheory for the relationship 
between pomset logic and CKA and illustrate this relationship with simple examples. 
In addition, to the usual classical or intuitionistic connectives --- both are possible 
--- the logic includes connectives that characterize both sequential and parallel 
composition.

In addition, we note that CKA allows programs with every possible
interleaving of parallel threads. However, to prove the correctness of
such programs, some restrictions must be imposed on what are the legal
interleavings. We provide a mechanism of `boxes' for this
purpose. Boxes identify protected parts of the system, so restricting
the possible interleavings. From the outside, one may interact with
the box as a whole, as if the program inside was atomic. On the other
hand, it is not possible to interact with its individual components,
as that would intuitively require opening the box. However, boxes can
be nested, with this atomicity observation holding at each level.
Pomset logic has context and box modalities that characterize this
situation.
\begin{note*}
  The term `Pomset logic' has already been used in work by
  Retor\'e~\cite{retore97}. We feel that reusing it does not introduce
  ambiguity, since the two frameworks arise in different contexts.
\end{note*}
\begin{example}[Running example: a distributed counter]\label{ex:distrib-count}
  \verbdef\pinstrprint{print(counter);}
  \verbdef\pinstrx{x:=0;}
  \verbdef\pinstry{y:=0;}
  \verbdef\pinstrxc{x:=counter;}
  \verbdef\pinstryc{y:=counter;}
  \verbdef\pinstrxincr{x:=x+1;}
  \verbdef\pinstryincr{y:=y+1;}
  \verbdef\pinstrcx{counter:=x;}
  \verbdef\pinstrcy{counter:=y;}
  \verbdef\instrprint{print(counter)}
  \verbdef\instrx{x:=0}
  \verbdef\instry{y:=0}
  \verbdef\instrxc{x:=counter}
  \verbdef\instryc{y:=counter}
  \verbdef\instrxincr{x:=x+1}
  \verbdef\instryincr{y:=y+1}
  \verbdef\instrcx{counter:=x}
  \verbdef\instrcy{counter:=y}
  \verbdef\instratomic£atomic{£
    \verbdef\instrclose£}£
  
  We consider here a program where a counter is incremented in
  parallel by two processes. The intention is that the counter should
  be incremented twice, once by each process. However, to do so each
  process has to first load the contents of the counter, then compute
  the increment, and finally commit the result to memory. A naive
  implementation is presented in \Cref{tab:ex-counter:code}.
  Graphically, we represent the print instruction
  \instrprint by $\Print$, the read instruction {\instrxc} by
  $\Read x$, the increment instruction {\instrxincr} by $\Compute x$,
  and finally the write instruction {\instrcx} by $\Write x$. We thus
  represent the previous program as displayed in
  \Cref{tab:ex-counter:pomset}.   
  \begin{figure}[t]
    \centering
    \noindent%
    \fbox{
      \begin{subfigure}[b]{.4\linewidth}\centering
        \begin{tabular}{l@{~}||@{~}l}
          \multicolumn{2}{c}{\pinstrprint}
          \\
          % \multicolumn{2}{c}{\pinstrx}
          % \\
          % \multicolumn{2}{c}{\pinstry}
          % \\
          \pinstrxc&\pinstryc
          \\
          \pinstrxincr&\pinstryincr
          \\
          \pinstrcx&\pinstrcy
          \\
          \multicolumn{2}{c}{\pinstrprint}
        \end{tabular}
        \caption{Pseudo code}
        \label{tab:ex-counter:code}
      \end{subfigure}\hspace{.05\linewidth}%
      \begin{subfigure}[b]{.5\linewidth}\centering
        \begin{tikzpicture}[xscale=1.5,yscale=.8]
          \node (I)at (0,0){$\Print$};
          \node (p1)at (1,0.5){$\Read x$};
          \node (p2)at (2,0.5){$\Compute x$};
          \node (p3)at (3,0.5) {$\Write x$};
          \node (q1)at (1,-0.5){$\Read y$};
          \node (q2)at (2,-0.5){$\Compute y$};
          \node (q3)at (3,-0.5) {$\Write y$};
          \node (F)at (4,0){$\Print$};
          \draw[thick,->,>=stealth](I)--(p1);
          \draw[thick,->,>=stealth](I)--(q1);
          \draw[thick,->,>=stealth](p1)--(p2);
          \draw[thick,->,>=stealth](p2)--(p3);
          \draw[thick,->,>=stealth](p3)--(F);
          \draw[thick,->,>=stealth](q1)--(q2);
          \draw[thick,->,>=stealth](q2)--(q3);
          \draw[thick,->,>=stealth](q3)--(F);
        \end{tikzpicture}
        \vspace{.4cm}
        \caption{Graphical representation}
        \label{tab:ex-counter:pomset}
      \end{subfigure}
    }
    \caption{Distributed counter}
    \label{tab:ex-counter}
  \end{figure}
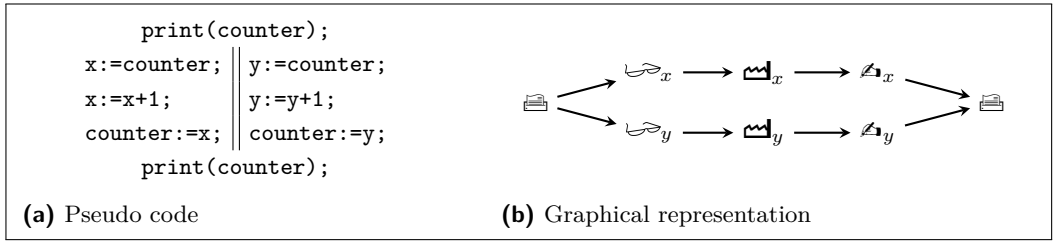

  \noindent%
  This program does not comply with our intended semantics, since the
  following run is possible:
  \begin{center}
    \begin{tikzpicture}[xscale=1.5,yscale=.8]
      \node (I)at (0,0){$\Print$};
      \node (p1)at (1,0){$\Read x$};
      \node (p2)at (3,0){$\Compute x$};
      \node (p3)at (5,0) {$\Write x$};
      \node (q1)at (2,0){$\Read y$};
      \node (q2)at (4,0){$\Compute y$};
      \node (q3)at (6,0) {$\Write y$};
      \node (F)at (7,0){$\Print$};
      \draw[thick,->,>=stealth](I)--(p1);
      \draw[thick,->,>=stealth](p1)--(q1);
      \draw[thick,->,>=stealth](q1)--(p2);
      \draw[thick,->,>=stealth](p2)--(q2);
      \draw[thick,->,>=stealth](q2)--(p3);
      \draw[thick,->,>=stealth](p3)--(q3);
      \draw[thick,->,>=stealth](q3)--(F);
    \end{tikzpicture}
  \end{center}
  The result is that the counter has been incremented by one. We can
  identify a subset of instructions that indicate there is a fault:
  the problem is that both read instructions happened before both
  write instructions; i.e., 
  \begin{center}
    \begin{tikzpicture}[yscale=.5]
      \node (p1)at (1,0.5){$\Read x$};
      \node (p3)at (3,0.5) {$\Write x$};
      \node (q1)at (1,-0.5){$\Read y$};
      \node (q3)at (3,-0.5){$\Write y$};
      \draw[thick,->,>=stealth](p1)--(p3);
      \draw[thick,->,>=stealth](p1)--(q3);
      \draw[thick,->,>=stealth](q1)--(p3);
      \draw[thick,->,>=stealth](q1)--(q3);
    \end{tikzpicture}
  \end{center}

  To preclude this problematic behaviour, a simple solution is to make the sequence
  `read;compute;write' \emph{atomic}. This yields the program in \Cref{tab:ex-atomic-counter:code}. Diagrammatically, this can
  be represented by drawing solid boxes around the
  {\instratomic\instrclose} blocks, as shown in
  \Cref{tab:ex-atomic-counter:pomset}.
  \begin{figure}[t]
    \centering
    \noindent%
    \fbox{
      \begin{subfigure}[b]{.4\linewidth}\centering
      \begin{tabular}{l@{~}||@{~}l}
        \multicolumn{2}{c}{\pinstrprint}
        \\
        % \multicolumn{2}{c}{\pinstrx}
        % \\
        % \multicolumn{2}{c}{\pinstry}
        % \\
        \instratomic&\instratomic
        \\
        ~~~\pinstrxc&~~~\pinstryc
        \\
        ~~~\pinstrxincr&~~~\pinstryincr
        \\
        ~~~\pinstrcx&~~~\pinstrcy
        \\
        \instrclose&\instrclose
        \\
        \multicolumn{2}{c}{\pinstrprint}
      \end{tabular}
      \caption{Pseudo code}
      \label{tab:ex-atomic-counter:code}
    \end{subfigure}\hspace{.05\linewidth}%
    \begin{subfigure}[b]{.5\linewidth}\centering
      \begin{tikzpicture}[xscale=1.5]
        \node (I)at (0,0){$\Print$};
        \node (p1)at (1,0.5){$\Read x$};
        \node (p2)at (2,0.5){$\Compute x$};
        \node (p3)at (3,0.5) {$\Write x$};
        \node (q1)at (1,-0.5){$\Read y$};
        \node (q2)at (2,-0.5){$\Compute y$};
        \node (q3)at (3,-0.5) {$\Write y$};
        \node (F)at (4,0){$\Print$};
        \draw[thick,->,>=stealth](I)--(p1);
        \draw[thick,->,>=stealth](I)--(q1);
        \draw[thick,->,>=stealth](p1)--(p2);
        \draw[thick,->,>=stealth](p2)--(p3);
        \draw[thick,->,>=stealth](p3)--(F);
        \draw[thick,->,>=stealth](q1)--(q2);
        \draw[thick,->,>=stealth](q2)--(q3);
        \draw[thick,->,>=stealth](q3)--(F);
        \node[draw,thick,rectangle,fit=(p1)(p2)(p3),inner sep=1mm](P){};
        \node[draw,thick,rectangle,fit=(q1)(q2)(q3),inner sep=1mm](Q){};
      \end{tikzpicture}
      \vspace{.6cm}
      \caption{Graphical representation}
      \label{tab:ex-atomic-counter:pomset}
    \end{subfigure}
    }
    \caption{Distributed counter with atomic increment}
    \label{tab:ex-atomic-counter}
  \end{figure}
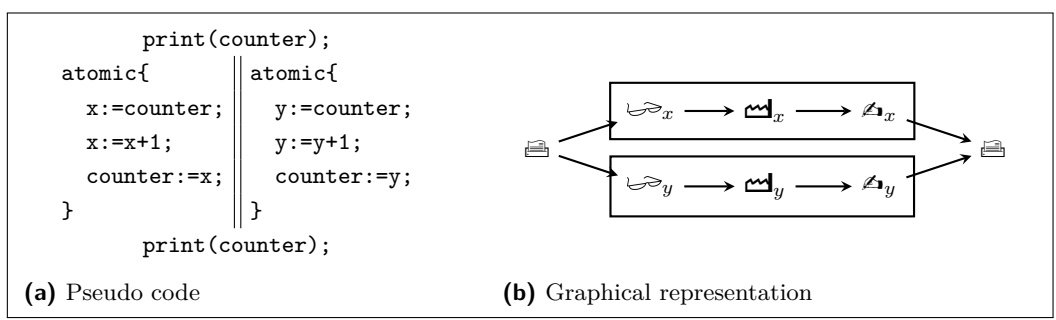
  This paper shows how to make these ideas formal.
\end{example}
%% todo : define f(X) where f:A->B and X:P(A), X:P(AxA), X:P(P(A))

% CONNECTIONS TO CSL-CKA PAPER. CONNECTIONS TO DEP LOGIC

%%% TODO : update the roadmap
In Section \ref{sec:algebra}, we extend pomsets with a new construct
for protection, namely boxes.  We provide a syntax for specifying such
pomsets and characterize precisely its expressivity. This enables us,
for example, to correctly represent the program from
Example~\ref{ex:distrib-count}. We present a sound and complete
axiomatization of these terms, with operators for boxing, sequential
and parallel composition, and non-deterministic choice, as well as the
constants \emph{abort} and \emph{skip}.

In Section \ref{sec:logic}, we introduce pomset logic. This logic
comes in both classical and intuitionistic variants. In addition to
the usual classical or intuitionistic connectives, this logic includes
connectives corresponding to each of sequential and parallel
composition. These two classes of connectives are combined to give the
overall logics, in the same way as the additives and multiplicatives
of BI (bunched implications logic) \cite{OP99,AP16,Pym2019}.  Just as
in BI and its associated separation logics \cite{OP99,IO01,Rey02},
pomset logic has both classical and intuitionistic variants. It also
includes modalities that characterize, respectively, protection, and
locality. These correspondences are made precise by
van~Benthem--Hennessy--Milner-type theorems asserting that two
programs are (operationally) equivalent iff they satisfy the same
formulae. We obtain such correspondences for several variants of our
framework.  In contrast to Hennessy--Milner logic, however, pomset
logic is a logic of \emph{behaviours} rather than of states and
transitions.

In Section \ref{sec:local}, we investigate local reasoning principles
for our logic of program behaviours. We showcase the possibilities of
our framework on an example. We conclude by briefly discussing future
work in Section~\ref{sec:future}.

\section{Algebra of Pomsets with Boxes} \label{sec:algebra}

In this section, we define our semantic model, and the corresponding syntax. 
We characterize the expressivity of the syntax, and axiomatize its equational theory.

Throughout this paper, we will use $\alphabet$ to denote a given set of atomic actions.

\subsection{Pomsets with boxes} 
\label{sec:def}

\subsubsection{Definitions and elementary properties}
\label{sec:basic-def}

\begin{definition}[Poset with boxes]
  A \emph{poset with boxes} is a tuple
  ${P\eqdef\tuple{\events[P],\order[P],\labels[P],\boxes[P]}}$, where
  $\events[P]$ is a finite set of \emph{events};
  $\order[P]\subseteq \events[P]\times\events[P]$ is a partial order;
  $\labels[P]:\events[P]\to\alphabet$ is a labelling function;
  $\boxes[P]\subseteq \pset{\events[P]}$ is a set of \emph{boxes},
  such that $\emptyset\notin\boxes[P]$.
\end{definition}
The partial order should be viewed as a set of necessary dependencies:
in any legal scheduling of the pomset, these dependencies have to be
satisfied. We therefore consider that a stronger ordering --- that is, one
containing more pairs --- yields a smaller pomset. The intuition is that
the set of legal schedulings of the smaller pomset is contained in
that of the larger one. The boxes are meant to further restrict the
legal schedulings: no event from outside a box may be interleaved
between the events inside the box. Subsequently, a pomset with more
boxes is smaller than one with less boxes. This ordering between
pomsets with boxes is formalized by the notion of homomorphism:
\begin{definition}[Poset morphisms]
  A \emph{(poset with boxes) homomorphism} is a map between event-sets
  that is bijective, label respecting, order preserving, and box
  preserving. In other words, a map $\phi:\events[P] \to \events[Q]$
  such that
  (i) $\phi$ is a bijection;
  (ii) $\labels[Q]\circ\phi=\labels[P]$;
  (iii) $\phi(\order[P])\subseteq\order[Q]$;
  (iv) $\phi(\boxes[P])\subseteq\boxes[Q]$.
  If in addition~(iii) holds as an equality, $\phi$ is
  called \emph{order-reflecting}. If on the other
  hand~(iv) holds as an equality $\phi$ is
  \emph{box-reflecting}. A homomorphism that is both order- and
  box-reflecting is a \emph{(poset with boxes) isomorphism}.
\end{definition}
In \Cref{fig:ex-subsumption} are some examples and a non-example of
subsumption between posets.
\begin{figure}[t]
  \centering
  \noindent%
  \fbox{  
    \begin{minipage}{.9\linewidth}\centering
      \begin{tikzpicture}[xscale=1]
        \node (I)at (7,0){$\Print$};
        \node (p1)at (8,0.5){$\Read x$};
        \node (p3)at (9,0.5) {$\Write x$};
        \node (q1)at (8,-0.5){$\Read y$};
        \node (q3)at (9,-0.5) {$\Write y$};
        \node (F)at (10,0){$\Print$};
        \draw[thick,->,>=stealth](I)--(p1);
        \draw[thick,->,>=stealth](I)--(q1);
        \draw[thick,->,>=stealth](p1)--(p3);
        \draw[thick,->,>=stealth](p3)--(F);
        \draw[thick,->,>=stealth](q1)--(q3);
        \draw[thick,->,>=stealth](q3)--(F);

        \node (I')at (0,0){$\Print$};
        \node (p1')at (1,0){$\Read x$};
        \node (p3')at (3,0) {$\Write x$};
        \node (q1')at (2,0){$\Read y$};
        \node (q3')at (4,0) {$\Write y$};
        \node (F')at (5,0){$\Print$};
        \draw[thick,->,>=stealth](I')--(p1');
        \draw[thick,->,>=stealth](p1')--(q1');
        \draw[thick,->,>=stealth](q1')--(p3');
        \draw[thick,->,>=stealth](p3')--(q3');
        \draw[thick,->,>=stealth](q3')--(F');

        \node()at($(I)!.5!(F')$){$\subsume$};
        \draw[thick,blue,->,>=stealth,dotted](I)to[out=160,in=20](I');
        \draw[thick,blue,->,>=stealth,dotted](F)to[out=-170,in=-10](F');
        \draw[thick,blue,->,>=stealth,dotted](p1)to[out=170,in=20](p1');
        \draw[thick,blue,->,>=stealth,dotted](p3)to[out=170,in=20](p3');
        \draw[thick,blue,->,>=stealth,dotted](q1)to[out=-170,in=-20](q1');
        \draw[thick,blue,->,>=stealth,dotted](q3)to[out=-170,in=-20](q3');
      \end{tikzpicture}
      \begin{tikzpicture}[xscale=1]
        \node (I)at (7.5,0){$\Print$};
        \node (p1)at (8.5,0.4){$\Read x$};
        \node (p3)at (9.5,0.4) {$\Write x$};
        \node (q1)at (8.5,-0.4){$\Read y$};
        \node (q3)at (9.5,-0.4) {$\Write y$};
        \node (F)at (10.5,0){$\Print$};
        \draw[thick,->,>=stealth](I)--(p1);
        \draw[thick,->,>=stealth](I)--(q1);
        \draw[thick,->,>=stealth](p1)--(p3);
        \draw[thick,->,>=stealth](p3)--(F);
        \draw[thick,->,>=stealth](q1)--(q3);
        \draw[thick,->,>=stealth](q3)--(F);
        \node[draw,thick,rectangle,fit=(p1)(p3),inner sep=1mm](P){};
        \node[draw,thick,rectangle,fit=(q1)(q3),inner sep=1mm](Q){};

        \node (I')at (0,0){$\Print$};
        \node (p1')at (1,0){$\Read x$};
        \node (p3')at (2,0) {$\Write x$};
        \node (q1')at (3.5,0){$\Read y$};
        \node (q3')at (4.5,0) {$\Write y$};
        \node (F')at (5.5,0){$\Print$};
        \draw[thick,->,>=stealth](I')--(p1');
        \draw[thick,->,>=stealth](p1')--(p3');
        \draw[thick,->,>=stealth](p3')--(q1');
        \draw[thick,->,>=stealth](q1')--(q3');
        \draw[thick,->,>=stealth](q3')--(F');
        \node[draw,thick,rectangle,fit=(p1')(p3'),inner sep=1mm](P'){};
        \node[draw,thick,rectangle,fit=(q1')(q3'),inner sep=1mm](Q'){};

        \node()at($(I)!.5!(F')$){$\subsume$};
        \draw[thick,blue,->,>=stealth,dotted](I)to[out=160,in=20](I');
        \draw[thick,blue,->,>=stealth,dotted](F)to[out=-170,in=-10](F');
        \draw[thick,blue,->,>=stealth,dotted](p1)to[out=170,in=20](p1');
        \draw[thick,blue,->,>=stealth,dotted](p3)to[out=170,in=20](p3');
        \draw[thick,blue,->,>=stealth,dotted](q1)to[out=-170,in=-20](q1');
        \draw[thick,blue,->,>=stealth,dotted](q3)to[out=-170,in=-20](q3');
      \end{tikzpicture}
      \begin{tikzpicture}[xscale=1]
        \node (I)at (7.5,0){$\Print$};
        \node (p1)at (8.5,0.4){$\Read x$};
        \node (p3)at (9.5,0.4) {$\Write x$};
        \node (q1)at (8.5,-0.4){$\Read y$};
        \node (q3)at (9.5,-0.4) {$\Write y$};
        \node (F)at (10.5,0){$\Print$};
        \draw[thick,->,>=stealth](I)--(p1);
        \draw[thick,->,>=stealth](I)--(q1);
        \draw[thick,->,>=stealth](p1)--(p3);
        \draw[thick,->,>=stealth](p3)--(F);
        \draw[thick,->,>=stealth](q1)--(q3);
        \draw[thick,->,>=stealth](q3)--(F);
        \node[draw,thick,rectangle,fit=(p1)(p3),inner sep=1mm](P){};
        \node[draw,thick,rectangle,fit=(q1)(q3),inner sep=1mm](Q){};

        \node (I')at (0,0){$\Print$};
        \node (p1')at (1,0){$\Read x$};
        \node (p3')at (3,0) {$\Write x$};
        \node (q1')at (2,0){$\Read y$};
        \node (q3')at (4,0) {$\Write y$};
        \node (F')at (5,0){$\Print$};
        \draw[thick,->,>=stealth](I')--(p1');
        \draw[thick,->,>=stealth](p1')--(q1');
        \draw[thick,->,>=stealth](q1')--(p3');
        \draw[thick,->,>=stealth](p3')--(q3');
        \draw[thick,->,>=stealth](q3')--(F');

        \node()at($(I)!.5!(F')$){$\not\subsume$};
      \end{tikzpicture}
    \end{minipage}
  }
  \caption{Poset subsumption}
  \label{fig:ex-subsumption}
\end{figure}
We introduce some notations. $\posets$ is the set of posets with
boxes. If $\phi$ is a homomorphism from $P$ to $Q$, we write
$\phi:P\homo Q$. If there exists such a homomorphism (respectively an
isomorphism) from $P$ to $Q$, we write $Q\subsume P$ (resp.
$Q\pomeq P$).
%% TODO: pomset hom example

\begin{lemma}
  ${\pomeq}$ is an equivalence relation. ${\subsume}$ is a partial order with respect to ${\pomeq}$.
\end{lemma}
\begin{remark}
  Note that the fact that ${\subsume}$ is antisymmetric with respect
  to ${\pomeq}$ relies on the finiteness of the posets considered
  here.%  Indeed, we can build infinite pomsets that are not isomorphic
  % but have nevertheless homomorphisms between them in both
  % directions.
  % 
  \begin{figure}[t]
    \centering
    \noindent%
    \fbox{
      \begin{tikzpicture}[xscale=2,yscale=.75]
        \node(P) at (0.5,2){$P:$};
        \node(p00) at (1,3){$\tuple{0,0}:a$};
        \node(p10) at (2,3){$\tuple{1,0}:a$};
        \node(p20) at (3,3){$\tuple{2,0}:a$};
        \node(p30) at (4,3){$\tuple{3,0}:a$};
        \node(pp0) at (5,3){$\cdots$};
        \node(p01) at (1,2){$\tuple{0,1}:a$};
        \node(p11) at (2,2){$\tuple{1,1}:a$};
        \node(p21) at (3,2){$\tuple{2,1}:a$};
        \node(p31) at (4,2){$\tuple{3,1}:a$};
        \node(pp1) at (5,2){$\cdots$};
        \node(p02) at (1,1){$\tuple{0,2}:a$};
        \node(p12) at (2,1){$\tuple{1,2}:a$};
        \node(p22) at (3,1){$\tuple{2,2}:a$};
        \node(p32) at (4,1){$\tuple{3,2}:a$};
        \node(pp2) at (5,1){$\cdots$};

        \draw[->,thick](p00) to (p10);
        \draw[->,thick](p10) to (p20);
        \draw[->,thick](p20) to (p30);
        \draw[->,thick](p30) to (pp0);
        
        \draw[->,thick](p01) to (p11);
        \draw[->,thick](p11) to (p21);
        \draw[->,thick](p21) to (p31);
        \draw[->,thick](p31) to (pp1);

        \draw[dashed](.5,0) to (5,0);
        
        \node(Q) at (0.5,-1.5){$Q:$};
        \node(q00) at (1,-1){$\tuple{0,0}:a$};
        \node(q10) at (2,-1){$\tuple{1,0}:a$};
        \node(q20) at (3,-1){$\tuple{2,0}:a$};
        \node(q30) at (4,-1){$\tuple{3,0}:a$};
        \node(qp0) at (5,-1){$\cdots$};
        \node(q01) at (1,-2){$\tuple{0,1}:a$};
        \node(q11) at (2,-2){$\tuple{1,1}:a$};
        \node(q21) at (3,-2){$\tuple{2,1}:a$};
        \node(q31) at (4,-2){$\tuple{3,1}:a$};
        \node(qp1) at (5,-2){$\cdots$};

        \draw[->,thick](q00) to (q10);
        \draw[->,thick](q10) to (q20);
        \draw[->,thick](q20) to (q30);
        \draw[->,thick](q30) to (qp0);
        
      \end{tikzpicture}
    }
    \caption{Example of mutual homomorphic posets that are not isomorphic}
    \label{fig:counterexample}
  \end{figure}
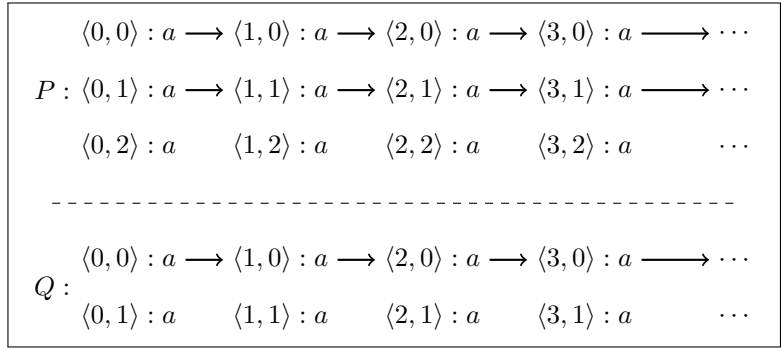
  For instance, consider the example is depicted in
  \Cref{fig:counterexample}. Formally, we fix some fixed symbol
  $a\in\alphabet$, and let $P$ and $Q$ be defined as follows:
  \begin{align*}
    P&\eqdef\tuple{\Nat\times\set{0,1,2},\order[P],\brack{\_\mapsto a},\emptyset}
    &\text{with}&\tuple{n,i}\order[P]\tuple{m,j}\eqdef \paren{n\leqslant m \wedge i=j<2}\\
    Q&\eqdef\tuple{\Nat\times\set{0,1},\order[Q],\brack{\_\mapsto a},\emptyset}
    &\text{with}&\tuple{n,i}\order[Q]\tuple{m,j}\eqdef \paren{n\leqslant m \wedge i=j=0}.
  \end{align*}
  One can plainly see that $P$ and $Q$ are not isomorphic, but there
  are indeed homomorphisms in both directions. Let us define the following functions:
  \begin{align*}
    \phi&:\events [P]\to\events [Q]&\psi&:\events[Q]\to\events[P]\\
        &\tuple{n,0}\mapsto\tuple{2n,0}&&\tuple{n,0}\mapsto{n,0}\\
        &\tuple{n,1}\mapsto\tuple{2n+1,0}&&\tuple{2n,1}\mapsto\tuple{n,1}\\
        &\tuple{n,2}\mapsto\tuple{n,1}&&\tuple{2n+1,1}\mapsto\tuple{n,2}.
  \end{align*}
  One may easily check that $\phi$ and $\psi$ are both
  poset-homomorphisms. 
\end{remark}

\begin{definition}[Pomsets with boxes]
  \emph{Pomsets with boxes} are equivalence classes of
  ${\pomeq}$. The set $\pomsets$ of pomsets with boxes is
  defined as $\posets/_{\pomeq}$.
\end{definition}%% move to remark

We now define some elementary poset-building operations.
\begin{definition}[Constants]
  Given a symbol $a\in\alphabet$, the \emph{atomic poset} associated
  with $a$ is defined as
  $\atom a\!\eqdef\!\tuple{\set{0},\brack{0\mapsto a},\Id{\set
      0},\emptyset}\!\in\!\posets$. The empty poset is defined as
  $\unitposet\!\eqdef\!\tuple{\emptyset,\emptyset,\emptyset,\emptyset}\!\in\!\posets$.
\end{definition}

\begin{remark}\label{rmk:nil}
  For any poset $P\in\posets$,
  $P\subsume \unitposet\Leftrightarrow P\revsubsume
  \unitposet\Leftrightarrow P\pomeq\unitposet$. This is because each
  of those relations imply there is a bijection between the events of
  $P$ and $\events[\unitposet]=\emptyset$. So we know that $P$ has no
  events, and since boxes cannot be empty, $P$ has no boxes
  either. Hence $P\pomeq \unitposet$.
\end{remark}

\begin{definition}[Compositions]
  Let $P,Q$ be two posets with boxes. The sequential composition $P\pomseq Q$
  and parallel composition $P\pompar Q$ are defined by:
  \begin{align*}
    P\pomseq Q
    &\eqdef\tuple{\events[P]\uplus\events[Q],
      \order[P]\cup\order[Q]\cup\paren{\events[P]\times\events[Q]},
      \labels[P]\sqcup\labels[Q], \boxes[P]\cup\boxes[Q]}\\
    P\pompar Q
    &\eqdef\tuple{\events[P]\uplus\events[Q],\order[P]\cup\order[Q],
      \labels[P]\sqcup\labels[Q], \boxes[P]\cup\boxes[Q]}, 
  \end{align*}
  where the symbol $\sqcup$ denotes the union of two functions; 
 that is, given $f:A\to C$ and $g:B\to C$, the function
  $f\sqcup g:A\uplus B\to C$ associates $f(a)$ to $a\in A$ and $g(b)$
  to $b\in B$.
\end{definition}
Intuitively, $P\pompar Q$ consists of disjoint copies of $P$ and $Q$
side by side. $P\pomseq Q$ also contains disjoint copies of $P$ and $Q$, but
also orders every event in $P$ before any event in $Q$.

\begin{definition}[Boxing]
  Given a poset $P$ its \emph{boxing} is denoted by $\boxing P$ and
  is defined by: $\boxing P\eqdef\tuple{\events[P],\order[P],\labels[P],\boxes[P]\cup\set{\events[P]}}$.
\end{definition}
Boxing a pomset simply amounts to drawing a box around it.

In our running example, the pattern of interest is a subset of the
events of the whole run. To capture this, we define the restriction of
a poset to a subset of its events. 

\begin{definition}[Restriction, sub-poset]
  For a given set of events $A\subseteq \events[P]$, we define the
  \emph{restriction of $P$ to $A$} as
  $\restrict P A\eqdef\tuple{A,\order[P]\cap\paren{A\times
      A},\restrict{\labels[P]}A,\boxes[P]\cap\pset A}$.
  We say that $P$ is a \emph{sub-poset} of $Q$, and write
  $P\issubpom Q$, if there is a set $A\subseteq \events[Q]$ such that
  $P\pomeq \restrict Q A$.
\end{definition}

Given a poset $P$, a set of events $A\subseteq\events[P]$ is called:
\begin{itemize}
\item\textbf{non-trivial} if $A\notin\set{\emptyset,\events}$.
\item\textbf{nested} if for any box $\beta\in\boxes[P]$ either $\beta\subseteq A$ or $A\cap\beta=\emptyset$;
\item\textbf{prefix} if for any $e\in A$ and $f\notin A$ we have $e\order[P] f$; and 
\item\textbf{isolated} if for any $e\in A$ and $f\notin A$ we have $e\not\order[P] f$ and $f\not\order[P] e$.
\end{itemize}

These properties characterize sub-posets of particular interest to
$P$. This is made explicit in the following observation:
\begin{fact}\label{fact:nested-sub-posets}
  Given a poset $P$ and a set of events $A\subseteq\events[P]$:
  \begin{enumerate}[(i)]
  \item $A$ is prefix and nested iff 
    $P\pomeq \restrict P A \pomseq \restrict P {\overline A}$;
  \item $A$ is isolated and nested iff
    $P\pomeq \restrict P A \pompar \restrict P {\overline A}$.
  \end{enumerate}
  (Here $\overline A$ denotes the complement of $A$ relative to
  $\events[P]$; that is, $\overline A\eqdef\events[P]\setminus A$.)
\end{fact}
This fact is very useful as a way to `reverse-engineer' how a poset
was built.
\subsubsection{Series--parallel pomsets}
\label{sec:sp-pom}

In the sequel, we will often restrict our attention to series--parallel 
pomsets. These are of particular interest since they are defined as 
those pomsets that can be generated from constants using the 
operators we have defined. 

\begin{definition}[Pomset terms, SP-Pomsets]
  A \emph{(pomset) term} is a syntactic expression generated from the following grammar:
  $s,t\in\spterms \Coloneqq 1 \Mid a \Mid s\tseq t\Mid s \tpar t\Mid \boxing s$.
  By convention $\tseq$ binds tighter than $\tpar$. A
  term is interpreted as a poset as follows:
  \begin{align*}
    \sem a&\eqdef \atom a
    &\sem 1&\eqdef \unitposet
    &\sem {\,\boxing s\,} &\eqdef \boxing{\,\sem s\,}\\
    \sem {s\tseq t}&\eqdef \sem s\pomseq \sem t
    &\sem{s\tpar t}&\eqdef \sem s \pompar \sem{t}. 
    & &
  \end{align*}
  A pomset $\brack P_{\pomeq}$ is called \emph{series--parallel} (or SP
  for short) if it is the interpretation of some term; that is, 
  $\exists s\in\spterms:\,\sem s\pomeq P$.
\end{definition}

\begin{example}
  The program in \Cref{tab:ex-counter} of the running example
  corresponds to
  \[\sem{\Print\tseq\paren{\Read x\tseq\Compute x\tseq\Write x\tpar\Read y\tseq\Compute y\tseq\Write y}\tseq\Print}.\]
  The corrected program, from \Cref{tab:ex-atomic-counter},
  corresponds to 
  \[\sem{\Print\tseq\paren{\boxing{\Read x\tseq\Compute x\tseq\Write
          x}\tpar\boxing{\Read y\tseq\Compute y\tseq\Write
          y}}\tseq\Print}.\]
  Finally, the problematic pattern we identified may be represented as
  $\sem{\paren{\Read x\tpar\Read y}\tseq\paren{\Write x\tpar\Write y}}$.
\end{example}

Series--parallel pomsets with boxes may also be defined by excluded
patterns, in the same style as the characterization of
series--parallel pomsets~\cite{vtl82,grabowski-1981,gischer88}. More
precisely, one can prove that a pomset $\brack P_{\pomeq}$ is
series--parallel iff and only if it does not contain any of the
patterns in \Cref{fig:sp-patterns}. Formally,
\newcommand\refpattern[1]{\ensuremath{\text{\bfseries P}_{\ref{pattern:#1}}}}
\begin{theorem}\label{thm:char-sp}
  A pomset $\brack P_{\pomeq}$ is series-parallel iff and only if it none of the following properties are satisfied:
  \begin{enumerate}[$\text{\bfseries P}_1:$]
    % label={$\text{\bfseries P}_\arabic*$:},ref={$\text{\bfseries P}_\arabic*$}
  \item\label{pattern:1} $\exists e_1,e_2,e_3,e_4\in \events[P]:\,
    {e_1\order[P] e_3} \wedge
    {e_2\order[P] e_3} \wedge
    {e_2\order[P] e_4} \wedge
    {e_1\not\order[P] e_4} \wedge
    {e_2\not\order[P] e_1} \wedge
    {e_4\not\order[P] e_3}$
  \item\label{pattern:2} $\exists e_1,e_2,e_3\in\events[P],\exists A,B\in\boxes[P]:\,
    {e_1\in A\setminus B} \wedge
    {e_2\in A\cap B} \wedge
    {e_3\in B\setminus A}$
  \item\label{pattern:3} $\exists e_1,e_2,e_3\in\events[P],\exists A\in\boxes[P]:\,
    {e_1\notin A} \wedge
    {e_2,e_3\in A} \wedge
    {e_1\order[P] e_2} \wedge
    {e_1\not\order[P] e_3}$
  \item\label{pattern:4} $\exists e_1,e_2,e_3\in\events[P],\exists A\in\boxes[P]:\,
    {e_1\notin A} \wedge
    {e_2,e_3\in A} \wedge
    {e_2\order[P] e_1} \wedge
    {e_3\not\order[P] e_1}$.
  \end{enumerate}
\end{theorem}

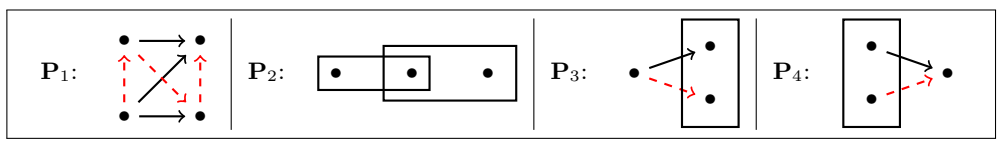
\begin{figure}[t]
  \centering
  \noindent%
  \fbox{
    \bgroup
    \begin{tabular}{rc|rc|rc|rc}
      \refpattern{1}:&
                        \begin{tikzpicture}[baseline=(m.base)]
                          \node(a)at(0,1){$\bullet$};
                          \node(b)at(0,0){$\bullet$};
                          \node(c)at(1,1){$\bullet$};
                          \node(d)at(1,0){$\bullet$};
                          \node(m)at($(a)!.5!(b)$){};
                          \draw[thick,->](a) to (c);
                          \draw[thick,->](b) to (c);
                          \draw[thick,->](b) to (d);
                          \draw[thick,->,dashed,red](b) to (a);
                          \draw[thick,->,dashed,red](d) to (c);
                          \draw[thick,->,dashed,red](a) to (d);
                        \end{tikzpicture}
      &
        \refpattern{2}:&
                          \begin{tikzpicture}[baseline=(a.base)]
                            \node(a)at(0,0){$\bullet$};
                            \node(b)at(1,0){$\bullet$};
                            \node(c)at(2,0){$\bullet$};
                            \node[fit=(a)(b),draw,thick,inner sep=1pt](box1){};
                            \node[fit=(c)(b),draw,thick,inner sep=5pt](box2){};
                          \end{tikzpicture}
      &
        \refpattern{3}:&
                          \begin{tikzpicture}[baseline=(a.base),yscale=.7]
                            \node(a)at(0,.5){$\bullet$};
                            \node(b)at(1,1){$\bullet$};
                            \node(c)at(1,0){$\bullet$};
                            \draw[thick,->](a) to (b);
                            \draw[thick,->,dashed,red](a) to (c);
                            \node[fit=(c)(b),draw,thick,inner sep=5pt](box2){};
                          \end{tikzpicture}
      &
        \refpattern{4}:&
                          \begin{tikzpicture}[baseline=(a.base),yscale=.7]
                            \node(a)at(1,.5){$\bullet$};
                            \node(b)at(0,1){$\bullet$};
                            \node(c)at(0,0){$\bullet$};
                            \draw[thick,->](b) to (a);
                            \draw[thick,->,dashed,red](c) to (a);
                            \node[fit=(c)(b),draw,thick,inner sep=5pt](box2){};
                          \end{tikzpicture}
    \end{tabular}
    \egroup
  }
  \caption{Forbidden patterns of SP-pomsets: dashed arrows (in {\color{red}red}) are negated.}
  \label{fig:sp-patterns}
\end{figure}

Before we discuss the proof of this result, we make a number of comments.

The four properties in \Cref{thm:char-sp} are invariant under
isomorphism, since they only use the ordering between events and the
membership of events to boxes. This is consistent with SP being a
property of pomsets, not posets.

Pattern \refpattern{1} is known as $N$, and is the forbidden pattern
of series-parallel pomsets (without boxes), as proved by
Gischer~\cite{gischer88}. Pattern \refpattern{2} indicates that the
boxes in an SP-pomset are well nested: two boxes are either disjoint,
or one is contained in the other.  Patterns \refpattern{3} and
\refpattern{4} reflect that an event is outside of a box cannot
distinguish the events inside by the order: it is either smaller than
all of them, larger than all of them, or incomparable with them.
  
Together, \refpattern{2}, \refpattern{3}, and \refpattern{4}
provide an alternative view of pomsets with boxes: one may see them as
\emph{hyper-pomsets}; that is, pomsets in which some events
(the boxes) can be labelled with non-empty pomsets (the contents of
the boxes). However, it seems that for our purposes the definition we provide
is more convenient. In particular, the definition of hyper-pomset
homomorphism is more involved.

Two auxiliary results on sub-posets, which we collect in the following
lemma, will be useful in this proof.
\begin{lemma}\label{lem:aux-sp-pom}
  Let $P$ be a poset with at least two events, such that
  $\events[P]\notin\boxes[P]$. Then:
  \begin{enumerate}[(i)]
  \item\label{claim:1} if $P$ contains a non-trivial prefix set, it
    contains one that is nested;
  \item\label{claim:2} if $P$ contains a non-trivial isolated set, it
    contains one that is nested.
  \end{enumerate}
\end{lemma}
\begin{proof}
  \begin{description}
  \item[(\ref{claim:1})]
    Assume there exists non-trivial prefix set. We pick a minimal one; 
    that is, a non-trivial prefix set $A$ such that for any non-trivial
    prefix $B$, if $B\subseteq A$, then $B=A$ (this is always possible
    since $\subseteq$ is a well-founded partial order on finite
    sets). If $A$ is nested, then $A$ satisfies our
    requirements. Otherwise, there is a box that is not contained in $A$
    while intersecting $A$. Since $P$ does not contain the
    pattern~\refpattern{2}, we know that we may pick a maximal such box
    $\beta$. This means that we know the following:
    \begin{mathpar}
      \beta \in\boxes[P]\and
      \paren{\forall\alpha\in\boxes[P],\,\beta\subseteq\alpha\Rightarrow\beta=\alpha}\and
      \exists e_1\in\beta\cap A\and
      \exists e_2\in\beta\setminus A.
    \end{mathpar}

    First, we show that $A\subseteq \beta$. Consider the set
    $A'\eqdef A\setminus \beta$. Clearly $A'\subsetneq A$ (since
    $e_1\in A\setminus A'$). We may also show that $A'$ is prefix. Let
    $e\in A'$ and $f\notin A'$. There are two cases:
    \begin{itemize}
    \item either $f\notin A$, then since $A$ is prefix and
      $e\in A'\subseteq A$ we have $e\order[P] f$;
    \item or $f\in A\cap \beta$. In this case, we use the fact that $P$
      does not have pattern~\refpattern{3}: we know that
      $e_2\in \beta\setminus A$, so since $e\in A'\subseteq A$ we have
      $e\order[P]e_2$, and since $e\notin\beta$ and $e_2,f\in\beta$ we may
      conclude that $e\order[P] f$.
    \end{itemize}
    Therefore, $A'$ is prefix and strictly contained in $A$. By
    minimality of $A$, $A'$ has to be trivial. Since $A$ is non-trivial
    this means that $A'=\emptyset$, hence that $A\subseteq\beta$.

    Now we know that $A\subseteq\beta$. Because we know that $\beta$ is
    not empty, and that $\events[P]\notin\boxes[P]$, we deduce that
    $\beta$ is non-trivial. Since $P$ does not contain
    pattern~\refpattern{2}, and by maximality of $\beta$, we know that
    $\beta$ is nested. We now conclude by showing that $\beta$ is in
    fact prefix. Let $e\in\beta$, and $f\notin\beta$. Since
    $A\subseteq\beta$ we get that $f\notin A$. By the prefix property of
    $A$ we get $e_1\order[P] f$, and since $e_1,e\in\beta$ and
    $f\notin\beta$, by the absence of pattern~\refpattern{4} we get
    that $e\order[P]f$.
  \item [(\ref{claim:2})] This proceeds in a similar manner. We pick a
    minimal non-trivial isolated set $A$, and try to find a maximal
    box $\beta$ such that $\beta\cap A\neq\emptyset$ and
    $\beta\setminus A\neq\emptyset$. If no such box exists, $A$ is
    already nested. If we do find such a box, we first show that $A$
    has to be contained in $\beta$. Then we use this to show that
    $\beta$ is a non-trivial nested isolated set.
  \end{description}
\end{proof}

\begin{proof}[Proof of~\Cref{thm:char-sp}]
  By a simple induction on terms, we can easily show that SP-posets
  avoid all four forbidden patterns. The more challenging direction is
  the converse: given a poset $P$ that does not contain any of the
  forbidden patterns, can we build a term $s\in\spterms$ such that
  $P\pomeq\sem s$. We construct this by induction on the size of $P$,
  defined as number of boxes plus the number of events. Notice that if
  a poset does not contain a pattern, then nor does any of its sub-posets.

  If $P$ has at most one event, then the following property holds:
  \begin{itemize}
  \item if $P$ has no events, then $P\pomeq\sem 1$;
  \item if $P$ has a single event $e$, let $a=\labels[P](e)$; we know that 
  	\[
		\boxes[P]\subseteq\setcompr{\beta\subseteq \set e}{\beta\neq\emptyset}=\set {\set e};
	\]
    \begin{itemize}
    \item if $\boxes[P]=\emptyset$ then $P\pomeq\atom a=\sem a$;
    \item if $\boxes[P]=\set{\set e}$ then $P\pomeq\boxing{\atom a}=\sem{\boxing a}$.
    \end{itemize}
  \end{itemize}

  If $\events[P]\in\boxes[P]$, then let $P'$ be the poset obtained by
  removing the box $\events[P]$. The size of $P'$ is strictly smaller
  than that of $P$, and $P\pomeq\boxing {P'}$. By induction, we get a
  term $s$ such that $\sem s\pomeq P'$, so $P\pomeq\sem{\boxing s}$.
  
  Consider now a pomset $P$ with at least two events, and such that
  $\events[P]\notin\boxes[P]$. As a corollary of Gischer's
  characterization theorem~\cite[Theorem 3.1]{gischer88}, we know that
  since $P$ is N-free (i.e., does not contain~\refpattern{1}) and
  contains at least two events, it contains either a non-trivial
  prefix set or a non-trivial isolated set.

  If $P$ contains a non-trivial, prefix set $A$, then by
  \Cref{lem:aux-sp-pom} $P$ contains a non-trivial, \emph{nested},
  prefix set $A'$. By \Cref{fact:nested-sub-posets}, this means that
  $P\pomeq \restrict P A\pomseq\restrict P{\overline A}$. We may thus
  conclude by induction.

  The case in which $P$ contains a non-trivial, isolated set $A$
  is handled similarly.
\end{proof}

\subsection{Sets of posets}

We now lift our operations and relations to sets of posets. This
allows us to enrich our syntax with a non-deterministic choice
operator.

\begin{definition}[Orderings on sets of posets]
  Let $A,B\subseteq\posets$, we define the following:
  \begin{description}
  \item[Isomorphic inclusion]: $A\subsetsim B$ iff $\forall P\in A,\,\exists Q\in B$ such that  $P\pomeq Q$
  \item[Isomorphic equivalence]: $A\pomeq B$ iff $A\subsetsim B\wedge B\subsetsim A$
  \item[Subsumption]: $A\subsume B$ iff $\forall P\in A,\,\exists Q\in B$ such that $P\subsume Q$. 
  \end{description}
\end{definition}
\begin{remark}
  Isomorphic inclusion and subsumption are partial orders with
  respect to isomorphic equivalence, which is an equivalence relation.
\end{remark}

\begin{definition}[Operations on sets of posets]
  We will use the set-theoretic union of sets of posets, as well as
  the pointwise liftings of the two products of posets and the boxing
  operators:
  \begin{align*}
    A\pomseq B&\eqdef\setcompr{P\pomseq Q}{\tuple{P,Q}\in A\times B}
    &\boxing A&\eqdef\setcompr{\boxing P}{P\in A}\\
    A\pompar B&\eqdef\setcompr{P\pompar Q}{\tuple{P,Q}\in A\times B}\!.
  \end{align*}
\end{definition}

\begin{definition}[Closure of a set of posets]
  The \emph{(downwards) closure} of a set of posets $S$ is the smallest set containing $S$
  that is downwards closed with respect to the subsumption order; that is,
  $\clpom S\eqdef\setcompr{P\in \posets}{\exists Q \in S:\,P\subsume Q}$.
  Similarly, the \emph{upwards closure} of $S$ is defined as:
  $\upclpom S\eqdef\setcompr{P\in \posets}{\exists Q \in S:\,P\revsubsume Q}$.
\end{definition}

\begin{remark}
  $\clpom{\paren{\_}}$ and $\upclpom{\paren{\_}}$ are Kuratowski
  closure operators~\cite{kuratowski1922}; i.e., they satisfy the
  following properties:
  \begin{mathpar}
    \clpom \emptyset=\emptyset\and
    A\subseteq \clpom A\and
    \clpom{\clpom A} = \clpom A\and
    \clpom{\paren{A\cup B}} = \clpom A\cup\clpom B.
  \end{mathpar}
  (And, similarly, for the upwards closure.) Using downwards-closures, we
  can express subsumption in terms of isomorphic inclusion:
  \begin{mathpar}
    A\subsume B\and
    \Leftrightarrow \and
    A\subsetsim \clpom B\and
    \Leftrightarrow \and
    \clpom A\subsetsim \clpom B.
  \end{mathpar}
  Similarly, the equivalence relation associated with $\subsume$,
  defined as the intersection of the relation and its converse,
  corresponds to the predicate $\clpom A\pomeq\clpom B$.
\end{remark}

\begin{definition}
  Terms are defined by the following grammar:
  \[
  	e,f\in \terms\Coloneqq 0 \;\Mid\; 1 \;\Mid\; a \;\Mid\; e\tseq f \;\Mid\; e \tpar f \;\Mid\; e \join f \;\Mid\; \boxing e. 
  \]
  Terms can be interpreted as finite sets of posets with boxes as follows:
  \begin{mathpar}
    \sem 0\eqdef \emptyset\and
    \sem 1 \eqdef\set\unitposet\and
    \sem a \eqdef\set{\atom a}\\
    \sem{\,\boxing e\,}\eqdef\boxing{\,\sem e\,}\and
    \sem {e\tseq f}\eqdef\sem e\pomseq\sem f\and
    \sem {e\join f}\eqdef\sem e\cup \sem f\and
    \sem {e\tpar f}\eqdef\sem e\pompar \sem f.
  \end{mathpar}
\end{definition}
\begin{remark}
  Interpreted as a program, $0$ represents failure: this is a program
  that aborts the whole execution. $\join$, on the other hand,
  represents non-deterministic choice. It can be used to model
  conditional branching. %%TODO: citation?
\end{remark}

\subsection{Axiomatic presentations of pomset algebra}
\label{sec:axioms}
\begin{table}[t]
  \centering
  \caption{Equational and inequational logic}
  \label{tab:eqlogic}
  \noindent%
  \fbox{
    \begin{minipage}{.95\linewidth}
      \begin{mathpar}
        \infer{e=f\in A}{A\vdash e=f}\and
        \infer{}{A\vdash e=e}\and
        \infer{A\vdash e = f}{A\vdash f = e}\and
        \infer{A\vdash e = f\and A\vdash f = g}{A\vdash e = g}\\
        \sigma,\tau:\alphabet\to\terms,\;
        \infer{\forall a\in \alphabet,\, A\vdash \sigma(a) = \tau (a)}{A\vdash \hat\sigma(e) = \hat\tau (e)}
      \end{mathpar}
      \hrule
      \begin{mathpar}
        \infer{e=f\in A}{A\vdash e\leq f}\and
        \infer{f=e\in A}{A\vdash e\leq f}\and
        \infer{e\leq f\in A}{A\vdash e\leq f}\and
        \infer{}{A\vdash e\leq e}\and
        \infer{A\vdash e \leq f\and A\vdash f \leq g}{A\vdash e \leq g}\\
        \sigma,\tau:\alphabet\to\terms,\;
        \infer{\forall a\in \alphabet,\, A\vdash \sigma(a) \leq \tau (a)}{A\vdash \hat\sigma(e) \leq \hat\tau (e)}
      \end{mathpar}
    \end{minipage}
  }
\end{table}
\begin{table}[t]
  \centering
  \caption{Axioms}
  \label{tab:axioms}
  \noindent%
  \fbox{
    \begin{minipage}{.95\linewidth}
      \begin{minipage}[t]{.45\linewidth}
        \begin{tableequations}
          \begin{align}
            s \tseq (t \tseq u) &= (s \tseq t) \tseq u \axlabel{bsp-eq-seq-ass}\\
            s \tpar (t \tpar u) &= (s \tpar t) \tpar u \axlabel{bsp-eq-par-ass}\\
            s \tpar t &= t \tpar s \axlabel{bsp-eq-par-comm}\\
            1 \tseq s &= s \axlabel{bsp-eq-seq-left-unit}\\
            s \tseq 1 &= s \axlabel{bsp-eq-seq-right-unit}\\
            1 \tpar s &= s \axlabel{bsp-eq-par-unit}\\
            \boxing{\boxing s} &= \boxing s \axlabel{bsp-eq-box-box}\\
            \boxing 1 &= 1  \axlabel{bsp-eq-box-unit}
          \end{align}
        \end{tableequations}
        \begin{tableequations}
          \begin{align}
            (s \tpar t) \tseq (u \tpar v)
            &\leq (s \tseq u) \tpar (t \tseq v) \axlabel{subs-exchange}\\
            \boxing s &\leq s \axlabel{subs-box-inf}
          \end{align}
        \end{tableequations}
      \end{minipage}\hfill%
      \begin{minipage}[t]{.45\linewidth}
        \begin{tableequations}
          \begin{align}
            e \join (f \join g) &= (e \join f) \join g \axlabel{bbsr-eq-plus-ass}\\
            e \join f &= f \join e \axlabel{bbsr-eq-plus-comm}\\
            e \join e &= e \axlabel{bbsr-eq-plus-idem}\\
            0 \join e &= e \axlabel{bbsr-eq-plus-unit}\\
            0 \tseq e &= e \tseq 0 =0 \axlabel{bbsr-eq-seq-absorbing}\\
            0 \tpar e &= 0 \axlabel{bbsr-eq-par-zero}\\
            e \tseq (f \join g) &= (e \tseq f) \join (e \tseq g)
                                  \axlabel{bbsr-eq-left-distr}\\
            (e \join f) \tseq g &= (e \tseq g) \join (f \tseq g)
                                  \axlabel{bbsr-eq-right-distr}\\
            e \tpar (f \join g) &= (e \tpar f) \join (e \tpar g)
                                  \axlabel{bbsr-eq-par-distr}\\
            \boxing 0 &= 0 \axlabel{bbsr-eq-box-zero}\\
            \boxing{e \join f} &= \boxing e \join \boxing f \axlabel{bbsr-eq-box-plus}
          \end{align}
        \end{tableequations}
      \end{minipage}\hfill%
    \end{minipage}}
\end{table}
We now introduce axioms to capture the various order and equivalence
relations we introduced over posets and sets of posets. Given a set of
axioms $A$ (i.e., universally quantified identities), we write
$A\vdash e=f$ to denote that the pair $\tuple{e,f}$ belongs to the
smallest congruence containing every axiom in $A$. Equivalently,
$A\vdash e=f$ holds iff this statement is derivable in the equational
logic described in \Cref{tab:eqlogic}. Similarly,
$A\vdash e\leq f$ is the smallest precongruence containing $A$, where
equality axioms are understood as pairs of inequational axioms. An
inference system is also provided in \Cref{tab:eqlogic}. We will
consider the following sets of axioms:
\begin{align*}
  \tag{Bimonoid with boxes}
  \axeq
  &\eqdef\eqref{bsp-eq-seq-ass}-\eqref{bsp-eq-box-unit}\\
  \tag{Concurrent monoid with boxes}
  \axinf
  &\eqdef\axeq,\eqref{subs-exchange}\eqref{subs-box-inf}\\
  \tag{Bisemiring with boxes}
  \axsr
  &\eqdef\axeq,\eqref{bbsr-eq-plus-ass}-\eqref{bbsr-eq-box-plus}\\
  \tag{Concurrent semiring with boxes} 
  \axsrinf
  &\eqdef\axsr,\eqref{subs-exchange},\eqref{subs-box-inf} 
\end{align*}
%
%\begin{remark}
In the last theory, inequational axioms $e\leq f$ should be read as
$e+f=f$. Indeed one can show that for $A\in\set{\axsr,\axsrinf}$, we
have 
\begin{mathpar}
  A\vdash e\leq f \Leftrightarrow A\vdash e+f=f\and
  A\vdash e = f \Leftrightarrow A\vdash e\leq f \wedge A\vdash f \leq e.
\end{mathpar}
% \end{remark}
% 

\subsubsection{Posets up to isomorphism: the free bimonoid with boxes}
\label{sec:bimon}

In this section we show that the axioms $\axeq$ provide a sound and complete axiomatisation of 
isomorphisms between posets with boxes; that is, 
\[
	\sem s \pomeq \sem t\Leftrightarrow \axeq\vdash s = t.
\]
Before we move to the prove of this statement, we need to make some
remarks and set up some auxiliary definitions.

First, notice that the constant $1$ can be handled easily:
\begin{lemma}\label{lem:eq1}
  For any term $t\in\spterms$, the set $\events[\sem t]$ is empty iff $\axeq\vdash t = 1$.
\end{lemma}
\begin{proof}
  Since $\events[\sem t]=\emptyset$, $t$ does not feature any symbol
  from $\alphabet$, meaning $t\in\spterms[\emptyset]$. We may then
  show that for every term $t\in\spterms[\emptyset]$ we have
  $\axeq\vdash t=1$.\qedhere
\end{proof}

We can also remove boxes from terms, in the following sense.
\begin{lemma}\label{lem:rmbox}
  Given a term $s\in\spterms$ such that 
  $\events[\sem s]\in\boxes[\sem s]$, there exists a term $t$ such that $\axeq\vdash s=\boxing t$ and $\events[\sem {t}]\notin\boxes[\sem {t}]$.
\end{lemma}
\begin{proof}
  We proceed by induction on $s$:
  \begin{itemize}
  \item$s=1$, $s=a$: contradicts the premiss;
  \item$s=\boxing {s'}$: we have to consider two cases:
    \begin{itemize}
    \item$\events[\sem{s'}]\in\boxes[\sem{s'}]$: in this case, by
      induction we get $t$ such that
      $\events[\sem{t}]\notin\boxes[\sem{t}]$, and
      $$\axeq\vdash s=\boxing{s'}=\boxing{\boxing{t}}=\boxing{t}$$
    \item$\events[\sem{s'}]\notin\boxes[\sem{s'}]$: pick $t=s'$; 
    \end{itemize}
  \item$s=s_1\tseq s_2$: we know the following facts:
    \begin{align*}
      &\events [\sem {s_1}]\cup\events[\sem {s_2}]=\events[\sem s]\in\boxes[\sem s]=\boxes[\sem {s_1}]\cup\boxes[\sem {s_2}]\\
      &\forall \beta\in\boxes[\sem {s_i}],\,\beta\subseteq\events[\sem{s_i}]\\
      &\events[\sem{s_1}]\cap\events[\sem{s_2}]=\emptyset .
    \end{align*}
    From them, we deduce that either $\events[\sem{s_1}]=\emptyset$ or
    $\events[\sem{s_2}]=\emptyset$. We conclude by applying the
    induction hypothesis to the appropriate sub-term, and use
    \Cref{lem:eq1} to conclude; 
  \item$s=s_1\tpar s_2$: same as $s_1\tseq s_2$.\qedhere
  \end{itemize}
\end{proof}

We now introduce a syntactic variant of the $\restrict P A$ operator
we introduced earlier, which extracts a sub-poset out of a poset,
guided by a subset of events.
\begin{definition}[Syntactic restriction]
  Let $s\in\spterms$ be a term and $A\subseteq\events[\sem s]$ a set
  of events. The \emph{syntactic restriction} of $s$ to $A$, written
  $\proj A s$ is defined by induction on terms as follows:
  \begin{align*}
    \proj A 1&\eqdef 1
    &
    \proj A a&\eqdef \left\{
               \begin{array}{ll}
                 1&if~A=\emptyset\\
                 a&otherwise
               \end{array}
                    \right.
    &
      \proj A {s\tseq t}
    &\eqdef \proj{A\cap\events[\sem s]}s\tseq\proj{A\cap\events[\sem s]} t\\
    &&\proj A {\boxing s}&\eqdef \left\{
                         \begin{array}{ll}
                           \boxing s&if~A=\events[\sem e]\\
                           \proj A s&otherwise
                         \end{array}
                                      \right.
    &
      \proj A {s\tpar t}
    &\eqdef \proj{A\cap\events[\sem s]}s\tpar\proj{A\cap\events[\sem s]} t.
  \end{align*}
  
\end{definition}

The main properties of this operator are stated in the following observation:
\begin{fact}\label{lem:proj}
  For any $s\in\spterms$ and $A\subseteq\events[\sem s]$ the following
  hold:
  \begin{enumerate}[(i)]
  \item $\sem{\proj A s}\pomeq\restrict{\sem s} A$;
  \item if $A$ is prefix and nested, then
    $\axeq\vdash\proj A s\tseq\proj{\overline A}s=s$;
  \item if $A$ is isolated and nested, then
    $\axeq\vdash\proj A s\tpar\proj{\overline A}s=s$.
  \end{enumerate}
\end{fact}

We may now establish the main result of this section:
\begin{theorem}\label{thm:bimon}
  For any pair of terms $s,t\in\spterms$, the following holds:
  $$\sem s \pomeq \sem t\Leftrightarrow \axeq\vdash s = t$$
\end{theorem}
\begin{proof}
  As often for this kind of result, the right-to-left implication; that is, 
  soundness, is very simple to check, by a simple induction on
  the derivation.

  For the converse direction, we prove the following statement by
  induction on the term $s$:
  \[\forall t\in \spterms,\, \sem s \pomeq \sem t \Rightarrow \axeq\vdash s=t.\]
  \begin{itemize}
  \item$s=1$: follows from \Cref{lem:eq1}.
  \item$s=a$: we prove the result by induction on $t$:
    \begin{itemize}
    \item$t=1$, $t=\boxing {t'}$: impossible since $\sem t\pomeq\atom a$;
    \item$t=b$: since $\sem t\pomeq\atom a$ this means $a=b$; i.e., by reflexivity $\axeq \vdash t=a$;
    \item$t=t_1\tseq t_2$: since $\atom a$ has a single event, and
      $\events[\sem{t_1\tseq t_2}] =
      \events[\sem{t_1}]\uplus\events[\sem{t_2}]$, that event must be
      either in $\events[\sem{t_1}]$ or in $\events[\sem{t_2}]$, and
      the other term has no event. The term that has no event, by
      \Cref{lem:eq1}, is provably equal to $1$. The term containing an
      event is isomorphic to $\atom a$, so by induction it is provably
      equal to $a$. Hence we get that $t$ is provably equal to either
      $a\tseq 1$ or $1\tseq a$, both of which are provably equal to
      $a$; 
    \item$t=t_1\tpar t_2$: same as $t_1\tseq t_2$.
    \end{itemize}
  \item$s=s_1\tseq s_2$: let
    $A\eqdef\events[\sem{s_1}]\subseteq \events[\sem s]$. Notice that
    $\overline A=\events[\sem
    s]\setminus\events[\sem{s_1}]=\events[\sem{s_2}]$.  Let $\phi$ be
    the isomorphism from $\sem s$ to $\sem t$ and let
    $t_1\eqdef\proj{\phi(A)}t$ and
    $t_2\eqdef \proj{\overline{\phi(A)}}t$. Since $A$ is nested and
    prefix, so is its image by the isomorphism $\phi$. Therefore by
    \Cref{lem:proj} we get that: $\axeq\vdash t=t_1\tseq t_2$.
    By properties of isomorphisms, we can also see that for any set
    $X\subseteq\events[\sem s]$ we have
    $\restrict{\sem s}X\pomeq\restrict{\sem t}{\phi(X)}$.  Hence we
    have
    $\sem{t_1}\pomeq\restrict{\sem t}{\phi(A)}\pomeq\restrict{\sem
      s}A\pomeq \sem{s_1}$. Similarly, and because by bijectivity of
    $\phi$ we have $\overline{\phi(A)}=\phi\paren{\overline A}$, we
    get $\sem {t_2}\pomeq\sem{s_2}$.  We may thus conclude by
    induction that $\axeq\vdash s_i=t_i$ ($i\in\set{1,2}$); that is, 
    \[\axeq\vdash s=s_1\tseq s_2=t_1\tseq t_2=t.\]
  \item$s=s_1\tpar s_2$: same as $s_1\tseq s_2$.\qedhere
  \end{itemize}
\end{proof}

\subsubsection{Posets up to subsumption: the free concurrent monoid}
\label{sec:concmon}

In this section we show a similar relation between the axioms of $\axinf$ and poset homomorphisms, namely that for any pair of terms $s,t$ we have $\sem s \subsume \sem t\Leftrightarrow \axinf\vdash s \leq t$.

To prove this result, we will deal with the two extra
axioms~\eqref{subs-exchange} and~\eqref{subs-box-inf} separately. In
order to do so, we define the following sub-orders of $\subsume$:
\begin{itemize}
\item we write $P\bsubsume Q$ if there is an
  order-reflecting poset homomorphism $\phi:Q\to P$;
\item we write $P\osubsume Q$ if there is an
  box-reflecting poset homomorphism $\phi:Q\to P$.
\end{itemize}
\begin{lemma}[Factorization of subsumption]\label{lem:fact-subs}
  If $P\subsume Q$, then there are $R_1,R_2$ such that:
  % \begin{mathpar}
  $P\osubsume R_1\bsubsume Q$ and $P\bsubsume R_2\osubsume Q$.
  In other words, $\subsume=\osubsume\circ\bsubsume$ and $\subsume=\bsubsume\circ\osubsume$.
\end{lemma}
\begin{proof}
  Let $\phi:Q\to P$ be a poset homomorphism witnessing $P\subsume Q$.
  We may define $R_1\eqdef\tuple{\events[Q],\order[Q],\labels[Q],\setcompr{B}{\phi\paren B\in\boxes[P]}}$ and $R_2\eqdef\tuple{\events[P],\order[P],\labels[P],\phi\paren {\boxes[Q]}}.$
  Checking that $P\osubsume R_1\bsubsume Q$ and
  $P\bsubsume R_2\osubsume Q$ hold is a simple matter of unfolding
  definitions.\qedhere
\end{proof}
We also notice the following properties of $\bsubsume$ and $\osubsume$
with respect to the forbidden patterns of series--parallel posets:
\begin{restatable}{lemma}{spsubsprop}\label{lem:sp-subs-prop}
  Let $P,Q$ be two posets:
  \begin{enumerate}[(i)]
  \item\label{item:p1-b} if $P\bsubsume Q$, then $P$ contains \refpattern{1} iff $Q$
    contains \refpattern{1};
  \item\label{item:p2-o} if $P\osubsume Q$, then $P$ contains \refpattern{2} iff $Q$
    contains \refpattern{2};
  \item\label{item:p3-b} if $P\bsubsume Q$ and $Q$ contains \refpattern{3}, then $P$
    contains \refpattern{3};
  \item\label{item:p4-b} if $P\bsubsume Q$ and $Q$ contains \refpattern{4}, then $P$
    contains \refpattern{4}.
  \end{enumerate}
\end{restatable}
\begin{proof}
  \begin{description}
  \item[(\ref{item:p1-b})] if $\phi:Q\to P$ is order reflecting, then by definition we have
    \begin{mathpar}
      x\order[Q]y\Leftrightarrow \phi(x)\order[P]\phi(y).
    \end{mathpar}
    The result follows immediately.
  \item[(\ref{item:p2-o})] if $\phi:Q\to P$ is order reflecting, then by definition we have
    \begin{mathpar}
      B\in\boxes[Q]\Leftrightarrow\phi\paren B\in\boxes[P]\\
      e\in B\setminus C\Leftrightarrow\phi(e)\in\phi\paren B\setminus\phi\paren C\and
      e\in B\cap C\Leftrightarrow\phi(e)\in\phi\paren B\cap\phi\paren C.
    \end{mathpar}
    The result follows immediately.
  \item[(\ref{item:p3-b})] if $\phi:Q\to P$ is order reflecting and $Q$ contains
    \refpattern{3} then by definition of the pattern we have
    $e_1,e_2,e_3\in\events[Q]$ and $B\in\boxes[Q]$ such that
    \begin{mathpar}
      e_1\notin B\and e_2\in B\and e_3\in B\and e_1\order[Q]e_2 \and e_1\not\order[Q] e_3.
    \end{mathpar}
    Since $\phi$ is a poset homomorphism, we know that
    $\phi\paren B\in\boxes[P]$ and $\phi(e_1)\order[P]\phi(e_2)$. By
    definition of the direct image, we also know that
    $\phi(e_1)\notin\phi\paren B$ and
    $\phi(e_2),\phi(e_3)\in\phi\paren B$. Finally, since $\phi$ is
    order reflecting $\phi(e_1)\not\order[P]\phi(e_3)$.
  \item[(\ref{item:p4-b})] similar to the proof for~(\ref{item:p3-b}).\qedhere
  \end{description}
\end{proof}

We now prove that $\axinf$ completely axiomatizes $\bsubsume$:
\begin{lemma}\label{lem:bsubsume}
  If $\sem s\bsubsume\sem t$, then $\axinf\vdash s\leq t$.
\end{lemma}
\begin{proof}
  \newcommand\lessboxPom[1]{\mathcal{H}\paren{#1}}
  \newcommand\lessboxTm[1]{\mathrm{H}\paren{#1}}
  First, we define the following operator $\lessboxPom\_:\posets\to\fpset\posets$:
  \[\lessboxPom P\eqdef\setcompr{\tuple{\events[P],\order[P],\labels[P],B}}{B\subseteq\boxes[P]}.\]
  From the definitions it is straightforward to show that
  $P\bsubsume Q$ iff $Q$ is isomorphic to some $P'\in\lessboxPom
  P$. This operator $\lessboxPom P$ can be mirrored on terms; i.e., we can associate inductively 
  to each term $s$ a finite set of terms $\lessboxTm s$ such that 
  % \begin{mathpar}
  %   \lessboxTm 1\eqdef\set 1\and
  %   \lessboxTm a\eqdef\set a\and
  %   \lessboxTm {\boxing s}\eqdef \lessboxTm s\cup\setcompr{\boxing{s'}}{s'\in\lessboxTm s}\and
  %   \lessboxTm{s\tseq t}\eqdef\setcompr{s'\tseq t'}{\tuple{s',t'}\in\lessboxTm s\times \lessboxTm t}\and
  %   \lessboxTm{s\tpar t}\eqdef\setcompr{s'\tpar t'}{\tuple{s',t'}\in\lessboxTm s\times \lessboxTm t}.
  % \end{mathpar}
  % Then we show by induction on $s$ that
  $\forall P\in\lessboxPom{\sem s},\exists t\in\lessboxTm s:\,\sem
  t\pomeq P$ and $\forall t\in\lessboxTm s,\,\axinf\vdash s\leq t$. We may therefore conclude:
  \begin{align*}
    \sem s\bsubsume\sem t&\Rightarrow\exists P\in\lessboxPom{\sem s}:\sem t\pomeq P\\
                         &\Rightarrow\exists t'\in\lessboxTm s:\sem {t'}\pomeq \sem t\\
                         &\Rightarrow\exists t'\in\lessboxTm s:\axeq\vdash t' = t \wedge \axinf\vdash s\leq t'\\
                         &\Rightarrow\axinf\vdash s\leq t'=t.\qedhere
  \end{align*}
\end{proof}
We now prove the same result for $\osubsume$. 
\begin{lemma}\label{lem:osubsume}
  If $\sem s\osubsume \sem t$ then $\axinf \vdash s\leq t$.
\end{lemma}
\begin{remark}
  The proof we give below relies on Gischer's completeness
  theorem~\cite{gischer88}. In the Coq proof however, we make no
  assumptions, and we do not have access to Gischer's
  result. Therefore we perform a different, more technically involved
  proof there, with Gischer's theorem as a corollary.
\end{remark}
\begin{proof}
  \newcommand\B{\mathcal B}%
  \newcommand\letter[1]{\ell\paren{#1}}%
  \newcommand\smap[1]{\mathpzc{s}\paren{#1}}%
  \newcommand\tmap[1]{\mathpzc{t}\paren{#1}}%
  We will perform this proof by induction on the number of boxes in
  $s$. Let $\phi:\sem t\to \sem s$ be the box-reflecting homomorphism
  witnessing $\sem s\osubsume\sem t$.

  If $s$ contains no boxes, then by Gischer's completeness theorem we know that
  \[\sem s\osubsume\sem t\Rightarrow \axeq,\eqref{subs-exchange}\vdash s\leq t.\]
  Hence, as we have $\axeq,\eqref{subs-exchange}\subseteq\axinf$, we
  get $\axinf\vdash s\leq t$.

  If on the other hand $s$ has boxes, consider the following set of
  boxes:
  \[\B\eqdef\setcompr{B\in\boxes[\sem s]}{\forall C\in\boxes[\sem
      s], B\subseteq C\Rightarrow B=C}\paren{=\max\boxes[\sem s]}.\]
  Notice that since $\phi$ is box-reflecting, we have that
  $\B=\phi\paren{\max\boxes[\sem t]}.$ Furthermore, we
  have the following property: for any box
  $B\in\max\boxes[\sem t]$, the map $\restrict \phi B$ is
  a box-reflecting homomorphism from $\restrict{\sem t}B$ to
  $\restrict{\sem s}{\phi\paren B}$. We pick new symbols for the
  elements of $\B$; that is, we find a set $\Sigma'$ disjoint from
  $\Sigma$ and a bijection $\letter\_:\B\to\Sigma'$. Using the
  observation~\Cref{lem:rmbox} we made earlier, we find two maps
  $\smap\_,\tmap\_:\Sigma'$ such that:
  \begin{align}
    \forall B\in\B,&&\axeq\vdash\proj B s&=\boxing{\smap{\letter B}}.\label{eq:3}\\
    \forall x\in\Sigma',&&\events[\sem{\smap x}]&\notin\boxes[\sem{\smap x}]\\
    \forall B\in\max\boxes[\sem t],&&\axeq\vdash\proj B t&=\boxing{\tmap{\letter {\phi(B)}}}.\label{eq:4}\\
    \forall x\in\Sigma',&&\events[\sem{\tmap x}]&\notin\boxes[\sem{\tmap x}] 
  \end{align}
  Clearly, for every $x\in\Sigma'$, $\smap x$ has strictly less
  boxes than $s$. Our previous observations imply that
  $\sem{\smap x}\osubsume\sem{\tmap x}$. Therefore, by induction
  hypothesis, we get that
  $\forall x\in\Sigma',\,\axinf\vdash\smap x\leq\tmap x$, hence:
  \[
  	\forall x\in\Sigma',\,\axinf\vdash\boxing{\smap x}\leq\boxing{\tmap x}.
  \]
  Combined with~\eqref{eq:3} and~\eqref{eq:4} this yields that 
  $\forall B\in\max\boxes[\sem t],\,\axinf\vdash\proj {\phi\paren B}s\leq\proj B t$.
  We define two substitutions
  $\sigma,\tau:\alphabet\cup\Sigma'\to\spterms$ as follows:
  \begin{mathpar}
    \sigma(x)\eqdef\left\{
      \begin{array}{ll}
        \proj B s&\text{ if }x=\letter B\\
        x&\text{ if }x\in\alphabet
      \end{array}
    \right.\and\tau(x)\eqdef\left\{
      \begin{array}{ll}
        \proj B t&\text{ if }x=\letter{\phi\paren{B}}\\
        x&\text{ if }x\in\alphabet
      \end{array}
    \right.
  \end{mathpar}
  Finally, we syntactically substitute maximal boxed sub-terms with
  letters from $\Sigma'$ in $s,t$, yielding terms
  $s',t'\in\spterms[\alphabet\cup\Sigma']$ such that
  $\hat\sigma\paren{s'}=s$, $\hat\tau\paren{t'}=t$, and neither $s'$
  nor $t'$ has any box. By unfolding the definitions, we can check
  that since $\sem s\osubsume\sem t$ we have
  $\sem s'\osubsume\sem t'$. Furthermore, since $s'$ and $t'$ do not
  contain any box, we may use Gischer's theorem to prove that
  $\axinf\vdash s'\leq t'$. We may now conclude: by applying
  $\sigma$ everywhere in the proof of $\axinf\vdash s'\leq t'$, we
  get that
  $\axinf\vdash s=\hat \sigma\paren{s'}\leq\hat\sigma\paren{t'}$.
  Since $\forall a\in\alphabet\cup\Sigma'$, we have
  $\axinf\vdash\sigma(a)\leq\tau(a)$, we get
  $\axinf\vdash \hat\sigma\paren{t'}\leq\hat\tau\paren{t'}=t$.\qedhere
\end{proof}

\begin{theorem}\label{thm:concmon}
  For any pair of terms $s,t\in\spterms$, the following holds:
  $$\sem s \subsume \sem t\Leftrightarrow \axinf\vdash s \leq t.$$
\end{theorem}
\begin{proof}
  Soundness is straightforward, and completeness arises from \Cref{lem:fact-subs},\Cref{lem:bsubsume}, and \Cref{lem:osubsume}.
\end{proof}

\subsubsection{Completeness results for sets of posets}
\label{sec:ax-sets-of-posets}

The following lemma allows us to extend seamlessly our completeness
theorem from $\axeq$ to $\axsr$ and from $\axinf$ to $\axsrinf$.
\begin{restatable}{lemma}{exprtoterm}\label{lem:expr-to-term}
  There exists a function $T_\_:\terms\to\fpset\spterms$ such that:
  $\axsr\vdash e = \bigjoin_{s\in T_e}s$ and 
  $\sem e \pomeq\setcompr{\sem s}{s\in T_e}$.
\end{restatable}
\begin{proof}
  $T_e$ is defined by induction on $e$:
  \begin{align*}
    T_0&\eqdef\emptyset
    &T_1&\eqdef \set 1\\
    T_a&\eqdef\set a
    &T_{\boxing e}&\eqdef \setcompr{\boxing s}{s\in T_e}\\
    T_{e\tseq f}&\eqdef\setcompr{s\tseq t}{\tuple{s,t}\in T_e\times T_f}
    &T_{e\tpar f}&\eqdef\setcompr{s\tpar t}{\tuple{s,t}\in T_e\times T_f}\\
    T_{e\join f}&\eqdef T_e\cup T_f.&&
  \end{align*}
  Checking the lemma is done by a simple induction.\qedhere
\end{proof}

From there, we can easily establish two completeness results. 
\begin{restatable}{theorem}{CompletenessSets}\label{thm:CompletenessSets}
  For any pair of terms $e,f\in \terms$, the following hold:
  \begin{align}
    \sem e\pomeq\sem f &\Leftrightarrow \axsr\vdash e = f\label{cor:bi-sr}\\
    \clpom{\sem e}\pomeq\clpom{\sem f} &\Leftrightarrow \axsrinf\vdash e = f \label{cor:conc-sr}
  \end{align}
\end{restatable}
\begin{proof}
  \begin{description}
  \item [(\ref{cor:bi-sr})]
    Soundness is easy to check. By a simple induction on the derivation
    tree, we can ensure that
    $\axsr\vdash e = f \Rightarrow \sem e \pomeq \sem f$. 
    
    \medskip 
    
    Using \Cref{lem:expr-to-term}, 
    we may rewrite any term $e$ as a finite union of series--parallel terms.
    Let $s\in T_e$. By soundness, there is $P\in \sem e$ such that
    $P\pomeq \sem s$. Since $\sem e \pomeq \sem f$, there is
    $Q\in \sem f$ such that $P\pomeq Q$. Since
    $\axsr\vdash f = \bigjoin_{s\in T_f}s$, by soundness there is
    $t\in T_f$ such that $Q\pomeq t$. Therefore $\sem s\pomeq \sem t$,
    so by \Cref{thm:bimon} we have $\axeq\vdash s = t$.  Since
    $\axeq\subseteq \axsr$, we also have $\axsr\vdash s = t$, and
    because $t\in T_f$ we get $\axsr\vdash s\leq f$. Since this holds for every $s\in T_e$, this means that 
    \[
    	\axsr\vdash e= \bigjoin_{s\in T_e}s\leq f.
    \]
    By a symmetric argument, we obtain $\axsr\vdash f\leq e$, allowing us
    to conclude by antisymmetry that $\axsr\vdash e = f$.
  \item[(\ref{cor:conc-sr})]
    Again, soundness is straightforward. For completeness, it is
    sufficient to show if
    $\sem e\subsume \sem f$ then $\axsrinf\vdash e\leq f$.
    Assume $\sem e\subsume\sem f$; i.e., every poset in~$\sem e$ is subsumed by some poset in $\sem f$. Since
    $\axsr\subseteq\axsrinf$, we get, by \Cref{lem:expr-to-term}, that 
    \begin{mathpar}
      \axsrinf\vdash e=\bigjoin_{s\in T_e}s \and
      \axsrinf\vdash f=\bigjoin_{t\in T_f}t.
    \end{mathpar}
    Let $s\in T_e$. Since $\sem e\subsume\sem f$, and because of
    \Cref{lem:expr-to-term}, we know that there is a
    term~${t\in T_f}$ such that $\sem s\subsume\sem t$. By
    \Cref{thm:concmon}, that means $\axinf\vdash s\leq t$. Since
    $\axinf\subseteq\axsrinf$, we get $\axsrinf\vdash s\leq t\leq
    f$. This means that for every~${s\in T_e}$, we have
    $\axsrinf\vdash s\leq f$, hence that
    $\axsrinf\vdash e=\bigjoin_{s\in T_e}s\leq f$.\qedhere
  \end{description}
\end{proof}

\section{Logic for pomsets with boxes} \label{sec:logic}

We introduce a logic for reasoning about pomsets with boxes, in the form of a bunched 
modal logic, in the sense of \cite{OP99,GMP05,CP09,AP16,Pym2019}, with substructural 
connectives corresponding to each of sequential and concurrent composition. Modalities 
characterize boxes and locality. The logic is also conceptually related to Concurrent 
Separation Logic \cite{CSL,BrookesO'Hearn2016}.
% Relationships 
% between Concurrent Separation Logic and CKA have bee explored in 
% \cite{OPVH2015}. 

In contrast with other work, pomset logic is a logic of
\emph{behaviours}. A behaviour is a run of some program, represented
as a pomset. The logic describes such behaviours in terms of the
order in which instructions are, or can be, executed, and the
separation properties of sub-runs. Note, in particular, that we do not
define any notion of \emph{state}.
On the contrary, existing approaches, such as dynamic logic and
Hennessy--Milner logic for example, put the emphasis on the state of
the machine before and after running the program. Typically, the
assertion language describes the memory-states, and some accessibility
relations between them. The semantics then relies on labelled
transition systems to interpret action modalities.

Here, the satisfaction relation (given in
\Cref{def:pomset-satisfaction}) directly defines a relation
between sets of behaviours and formulas. An intuitionistic version of
the semantics given in \Cref{def:pomset-satisfaction} might
be set up --- cf. Tarski's semantics and the semantics of relevant logic 
--- in terms of (ternary) relations on behaviours. 

\subsection{Pomset logic: definitions}
\label{sec:pom-logic}
%% TODO:
%% - examples
%% - v:Prop -> P(\Sigma)
%% - Jup  Jdown
%% - negation K
%% - forall <-> not exists not

We generate the set of formulas $\Kformulas$ and the set of positive formulas $\Jformulas$ as follows:
\begin{align*}
  \phi,\psi\in \Jformulas&\Coloneqq
                           \emptyprop
                           \Mid a
                           \Mid \phi\vee \psi
                           \Mid \phi\wedge\psi
                           \Mid \phi\then\psi
                           \Mid \phi \nextto \psi
                           \Mid \boxing \phi
                           \Mid \context \phi\\
  \phi,\psi\in \Kformulas&\Coloneqq
                           \emptyprop
                           \Mid a
                           \Mid \phi\vee \psi
                           \Mid \phi\wedge\psi
                           \Mid \phi\then\psi
                           \Mid \phi \nextto \psi
                           \Mid \boxing \phi
                           \Mid \context \phi
                           \Mid \neg \phi
\end{align*}\vspace{-5mm}
\begin{remark}
  Here the atomic predicates are chosen to be exactly
  $\alphabet$. Another natural choice would be a separate set $Prop$
  of atomic predicates, together with a valuation
  $v:Prop\to\pset\alphabet$ to indicate which actions satisfy which
  predicate. Both definitions are equivalent:
  \begin{itemize}
  \item to encode a formula over $Prop$ as a formula over $\alphabet$,
    simply replace every predicate $p\in Prop$ with the formula
    $\bigvee_{a\in v(p)}a$
  \item to encode a formula over $\alphabet$ as one over $Prop$, we
    need to make the customary assumption that
    $\forall a\in\alphabet,\,\exists p\in Prop:\,v(p)=\set a$.
  \end{itemize}
\end{remark}

These formulas are interpreted over posets. We define a satisfaction
relation $\models_R$ that is parametrized by a relation
$R\subseteq\posets\times\posets$ (to be instantiated later on with $\pomeq$, $\subsume$, and
$\revsubsume$). % We will denote the resulting satisfaction relations as
% \begin{mathpar}
%   \satK\eqdef \models_\equiv\and
%   \satJd\eqdef\models_\subsume\and
%   \satJu\eqdef\models_\revsubsume.
% \end{mathpar}
\begin{definition}\label{def:pomset-satisfaction}
$P\models_R\phi$ is defined by induction on $\phi\in\Kformulas$:
\begin{itemize}
\item $P\models_R \emptyprop$ iff $R\paren{P,\unitposet}$
\item $P\models_R a$ iff $R\paren{P,\atom a}$
\item $P\models_R \neg\phi$ iff $P\not\models_R\phi$
\item $P\models_R \phi\vee \psi$ iff $P\models_R \phi$ or $P\models_R \psi$
\item $P\models_R \phi\wedge \psi$ iff $P\models_R \phi$ and $P\models_R \psi$
\item $P\models_R \phi\then\psi$ iff $\exists P_1,P_2$ such that $R\paren{P, P_1\pomseq P_2}$ and $P_1\models_R \phi$ and $P_2\models_R \psi$
\item $P\models_R \phi \nextto \psi$ iff $\exists P_1,P_2$ such that $R\paren{P, P_1\pompar P_2}$ and $P_1\models_R \phi$ and $P_2\models_R \psi$
\item $P\models_R \boxing \phi$ iff $\exists Q$ such that $R\paren{P, \boxing {Q}}$ and $Q\models_R \phi$
\item $P\models_R \context \phi$ iff $\exists P',Q$ such that $R\paren{P,P'}$ and $P'\contains Q$ and $Q \models_R \phi$.
\end{itemize}
\end{definition}

\noindent The operator $\boxing -$ describes the (encapsulated) properties of boxed terms. The operator 
$\context -$ identifies a property of a term that is obtained by removing parts, including boxes 
and events, of its satisfying term (i.e., its world) such that remainder satisfies the formula that it guards.  
The meanings of these operators are discussed more fully in \Cref{sec:frame}.

Note that $\satJu$ and $\satJd$ will only be used with positive
formulas. Given a formula $\phi$ and a relation $R$, we may define the
$R$-semantics of $\phi$ as
$\psem[R]\phi\eqdef\setcompr{P\in\posets}{P\models_R\phi}$.

\begin{example}
  Recall the problematic pattern we saw in the running example; that is,  
  \begin{center}
    \begin{tikzpicture}[yscale=.8]
      \node (p1)at (1,0.5){$\Read x$};
      \node (p3)at (3,0.5) {$\Write x$};
      \node (q1)at (1,-0.5){$\Read y$};
      \node (q3)at (3,-0.5){$\Write y$};
      \draw[thick,->,>=stealth](p1)--(p3);
      \draw[thick,->,>=stealth](p1)--(q3);
      \draw[thick,->,>=stealth](q1)--(p3);
      \draw[thick,->,>=stealth](q1)--(q3);
    \end{tikzpicture}
  \end{center}
  This pattern can be represented by the formula
  $\conflict{}\eqdef\context{\paren{\Read x\nextto\Read y}\then\paren{\Write x\nextto \Write y}}$.
\end{example}

We may also interpret these formulas over sets of posets. We consider here two ways a set of posets $X$ may satisfy a formula:
\begin{itemize}
\item $X$ satisfies $\phi$ \emph{universally} if every poset in $X$ satisfies
  $\phi$;
\item $X$ satisfies $\phi$ \emph{existentially} if some poset in $X$ satisfies
  $\phi$.
\end{itemize}

Combined with our three satisfaction relations for pomsets,
this yields six definitions:
\begin{align*}
  X\satKU \phi&\text{ iff }\forall P\in X,P\satK\phi
  &X\satKE \phi&\text{ iff }\exists P\in X,P\satK\phi\\
  X\satJuU \phi&\text{ iff }\forall P\in X,P\satJu\phi
  &X\satJuE \phi&\text{ iff }\exists P\in X,P\satJu\phi\\
  X\satJdU \phi&\text{ iff }\forall P\in X,P\satJd\phi
  &X\satJdE \phi&\text{ iff }\exists P\in X,P\satJd\phi.
\end{align*}
% We write $e\models_x^y\phi$, with $\tuple{e,\phi,x,y}\in\terms\times\Kformulas\times\set{\Ju,\Jd,K}\times\set{\forall,\exists}$,
% to mean $\sem e\models_x^y\phi$.
For a term $e\in\terms$, we write $e\models_R^y\phi$ to mean
$\sem e\models_R^y\phi$. In terms of $R$-semantics, these definitions
may be formalized as 
\begin{align}
  e\models_R^\exists\phi&\Leftrightarrow \sem e\cap\psem[R]\phi\neq\emptyset
  	\quad\mbox{\rm and}\quad 
  &e\models_R^\forall\phi&\Leftrightarrow \sem e\subseteq\psem[R]\phi.\label{eq:sat}
\end{align}

\subsection{Properties of pomset logic}
\label{sec:facts-logic}

We now discuss some of the properties of pomset logic. First, notice that if the relation $R$ is transitive, 
then for any posets $P,Q$ and any formula $\phi\in\Jformulas$, we have that 
\begin{equation} 
  \label{eq:R-closure}
  R\paren{P,Q} \text{ and }Q\models_R\phi\Rightarrow P\models_R\phi.
\end{equation}
\begin{proof}
  Let $R$ be a transitive relation, we show by induction on
  $\Phi\in\Jformulas$ that $\forall P,Q\in\posets$, if $R\paren{P,Q}$
  and $Q\models_R\Phi$ then $P\models_R\Phi$.
  \begin{itemize}
  \item If $\Phi=\phi\vee \psi$ or $\Phi=\phi\wedge \psi$, we use the
    induction hypothesis to show that $Q\models_R \Phi$ implies
    $P\models_R\Phi$.
  \item If $\Phi=\emptyprop$, $\Phi=a$, $\Phi=\phi\then\psi$,
    $\Phi=\phi\nextto\psi$, $\Phi=\boxing\phi$, or
    $\Phi=\context \phi$, then the satisfaction relation says
    $Q\models_R\Phi$ iff $R\paren{Q,Q'}$ and $h(Q')$. Since
    $R\paren{P,Q}$ and $R\paren{Q,Q'}$, by transitivity of $R$ we get that
    $R\paren{P,Q'}$ and thus we conclude that $P\models_R\Phi$
    without using the induction hypothesis.\qedhere
  \end{itemize}  
\end{proof}

If, additionally, $R$ is symmetric, this property may be strengthened to 
\begin{equation}
  \label{eq:R-sym-closure}
  \forall P,Q\in\posets,\,\forall \phi\in\Kformulas,\,\text{if }R\paren{P,Q}\text{, then } 
  P\models_R\phi\Leftrightarrow Q\models_R\phi.
\end{equation}
\begin{proof}
  Let $R$ be a symmetric and transitive relation, we show by induction
  on $\Phi\in\Kformulas$ that $\forall P,Q\in\posets$, if
  $R\paren{P,Q}$ then $P\models_R\Phi$ iff $Q\models_R\Phi$.
  \begin{itemize}
  \item If $\Phi=\neg\phi$, $\Phi=\phi\vee \psi$, or $\Phi=\phi\wedge
    \psi$, we use the induction hypothesis to conclude.
  \item If $\Phi=\emptyprop$, $\Phi=a$, $\Phi=\phi\then\psi$,
    $\Phi=\phi\nextto\psi$, $\Phi=\boxing\phi$, or
    $\Phi=\context \phi$, then the satisfaction relation says
    $Q\models_R\Phi$ iff $R\paren{Q,Q'}$ and $h(Q')$. Since
    $R\paren{P,Q}$, by symmetry and transitivity of $R$ we get that
    $R\paren{P,Q'}\Leftrightarrow R\paren{Q,Q'}$ and thus we
    conclude that $P\models_R\Phi\Leftrightarrow Q\models_R \Phi$
    without using the induction hypothesis.\qedhere
  \end{itemize}
\end{proof}

Furthermore, increasing the relation $R$ increases the satisfaction relation as well:
\begin{equation}
  \label{eq:R-extension}
  R\subseteq R'\Rightarrow\forall \phi\in\Jformulas,\forall P\in\posets,\,
  P\models_R\phi\Rightarrow P\models_{R'}\phi.
\end{equation}
\begin{proof}
This follows by a straightforward induction on formulas.
\end{proof}

From these observations and~\eqref{eq:sat}, we obtain the following
characterizations of the universal satisfaction relations for
$R\in\set{\pomeq,\subsume,\revsubsume}$:
\begin{align}
  e\satKU\phi&\Leftrightarrow \sem e\subsetsim\psem[\K]\phi\label{eq:satKU}\\
  e\satJdU\phi&\Leftrightarrow \sem e\subsume\psem[\Jd]\phi\label{eq:satJdU}\\
  e\satJuU\phi&\Leftrightarrow\forall P\in\sem e,\,\exists Q\in\psem[\Ju]\phi:P\revsubsume Q. \label{eq:satJuU}
\end{align}

Additionally, the following preservation properties hold for sets of posets:
\begin{align}
  e\subsetsim f\Rightarrow \forall\phi\in\Kformulas,\,& \paren{e\satKE \phi \Rightarrow f\satKE \phi}
\wedge
         \paren{f\satKU \phi \Rightarrow e\satKU \phi}\label{eq:semK}\\
  e\subsume f\Rightarrow\forall\phi\in\Jformulas,\,&\paren{ e\satJuE \phi \Rightarrow f\satJuE \phi}
  \wedge\paren{f\satJdU \phi \Rightarrow e\satJdU \phi} \label{eq:semJ}
\end{align}
\begin{proof}
  \begin{description}
  \item[(\ref{eq:semK})] Assume $e\subsetsim f$.
    \begin{itemize}
    \item If $e\satKE \phi$, then there exists a poset
      $P\in\sem e\cap\psem[\K]\phi$. Because $e\subsetsim f$, we can
      find a poset $Q\in\sem f$ such that $P\pomeq Q$. Since
      $P\in\psem[\K]\phi$ and $\psem[\K]\phi$ is closed under $\pomeq$
      we get $Q\in\sem f\cap\psem[\K]\phi$.

    \item If $f\satKU \phi$, then $\sem f\subsetsim\psem[\K]\phi$. By
      transitivity, we get that
      $\sem e\subsetsim f\subsetsim\psem[\K]\phi$; i.e.,
      $e\satKU\phi$.
    \end{itemize}
  \item[(\ref{eq:semJ})] Assume $e\subsume f$.
    \begin{itemize}
    \item If $e\satJuE \phi$, then there exists a poset
      $P\in\sem e\cap\psem[\Ju]\phi$. Because $e\subsume f$, we can
      find a poset $Q\in\sem f$ such that $P\subsume Q$. Since
      $P\in\psem[\Ju]\phi$ and $\psem[\Ju]\phi$ is upwards-closed we
      get $Q\in\sem f\cap\psem[\Ju]\phi$.
      
    \item If $f\satJdU \phi$, then $\sem f\subsume\psem[\Jd]\phi$.  By
      transitivity, we get that
      $\sem e\subsume\sem f\subsume\psem[\Jd]\phi$,
      that is, $e\satJdU\phi$.\qedhere
    \end{itemize}
  \end{description}
\end{proof}

We can build formulas from series--parallel terms: $\phi(a)\eqdef a$,
$\phi(1)\eqdef \emptyprop$,
$\phi\paren{\boxing s}\eqdef\boxing{\phi(s)}$,
$\phi(s\tseq t)\eqdef \phi(s)\then\phi(t)$, and
$\phi(s \tpar t)\eqdef \phi(s)\nextto \phi(t)$. Using $T_\_$, we
generalize this construction our full syntax: given a term
$e\in\terms$, we define the formula
$\Phi(e)\eqdef\bigvee_{s\in T_e}\phi(s)$.  These formulas are closely
related to terms thanks to the following lemma:

\begin{restatable}{lemma}{termtoform}\label{lem:term-to-form}
  For any term $s\in\spterms $ and any poset $P$, we have that 
  \begin{mathpar}
    P\satK \phi (s) \Leftrightarrow P \pomeq \sem s\and
    P\satJu \phi (s) \Leftrightarrow P \revsubsume \sem s\and
    P\satJd \phi (s) \Leftrightarrow P \subsume \sem s.
  \end{mathpar}
  For a term $e\in\terms$ and a set of posets $X\subseteq\posets$, we
  have that 
  \begin{mathpar}
    X\satKU\Phi(e)\Leftrightarrow X\subsetsim\sem e\and
    X\satJdU\Phi(e)\Leftrightarrow X\subsume\sem e.
  \end{mathpar}
\end{restatable}
\noindent%
The proof of this lemma may be found in the appendix.

% We may
% rephrase the results we have established so far from this perspective:
% \begin{description}
% \item[Lemma~\ref{lem:K-impl-J}]: $\psem[\K]\phi\subseteq\psem[\Ju]\phi\cap\psem[\Jd]\phi$.
% \item[Lemma~\ref{lem:K-equiv}]: $\psem[\K]\phi$ is closed under $\pomeq$.
% \item[Lemma~\ref{lem:J-subs}]:
%   $\psem[\Ju]\phi=\upclpom{\psem[\Ju]\phi}$ and
%   $\psem[\Jd]\phi=\clpom{\psem[\Jd]\phi}$; i.e., the $\Ju$ semantics is
%   upwards-closed and the $\Jd$-semantics is downwards-closed.
% \item[Lemma~\ref{lem:term-to-form}]: $\psem[\K]{\phi(s)}\pomeq\set{\sem s}$, $\psem[\Jd]{\phi(s)}\pomeq\clpom{\set{\sem s}}$, $\psem[\Ju]{\phi(s)}\pomeq\upclpom{\set{\sem s}}$.
% \end{description}
As an immediate corollary, for any $e\in\terms$ and any
$s\in\spterms$, we obtain that:
\begin{align}
  \label{eq:phi}
    e\satKE\phi(s)\Leftrightarrow& \sem s\in\sem e
  &e\satJuE\phi(s)\Leftrightarrow& \sem s\in\clpom{\sem e}
\end{align}
\begin{proof}
  \begin{align*}
  e\satKE\phi(s)\Leftrightarrow&\sem e
  \cap\psem[\K]{\phi(s)}\neq\emptyset\Leftrightarrow\sem e
  \cap\set{\sem s}\neq\emptyset\Leftrightarrow\sem s\in\sem e.\\
  e\satJuE\phi(s)\Leftrightarrow&\sem e
  \cap\psem[\Ju]{\phi(s)}\neq\emptyset\Leftrightarrow\sem e
  \cap\upclpom{\set{\sem s}}\neq\emptyset\Leftrightarrow\sem
  s\in\clpom{\sem e}.\qedhere
  \end{align*}
\end{proof}

We can now establish adequacy lemmas. These should be understood as appropriate formulations 
of the completeness theorems relating operational equivalence and logical equivalence in the sense 
of van~Benthem~\cite{vB14} and Hennessy--Milner~\cite{HP80,Milner89} for this logic (cf.~\cite{AP16}).  
From the results we have established so far, we may directly prove the following:

\begin{restatable}{proposition}{adequacyPom}\label{lem:adequacy}
  For a pair of series--parallel terms $s,t\in\spterms$,
  \begin{align}
    \axeq\vdash s = t
    &\Leftrightarrow\forall\phi\in\Kformulas,\,\paren{\sem s
      \satK \phi\Leftrightarrow\sem t \satK \phi} 
      \label{eq:adequacy-K-pomset}\\
    \axinf\vdash s \leq t
    &\Leftrightarrow\forall\phi\in\Jformulas,\,
    \paren{\sem s \satJu \phi\Rightarrow\sem t \satJu \phi}. 
      \label{eq:adequacy-J-pomset}
  \end{align}
\end{restatable}
\begin{proof}
  \begin{align*}
    \axeq\vdash s = t
    &\Leftrightarrow \sem s\pomeq \sem t\tag{\Cref{thm:bimon}}\\
    &\Rightarrow\forall\phi\in\Kformulas,\,
      \sem s \satK \phi \Leftrightarrow \sem t \satK \phi.
      \tag{\Cref{eq:R-sym-closure}}\\
    &\Rightarrow\sem s\satK\phi(t)\tag{Since $\sem t\satK\phi(t)$ follows from \Cref{lem:term-to-form}}\\
    &\Leftrightarrow\sem s\pomeq\sem t\tag{\Cref{lem:term-to-form}}\\
    &\Leftrightarrow \axeq\vdash s = t \tag{\Cref{thm:bimon}}
  \end{align*}
  \begin{align*}
    \axinf\vdash s \leq t
    &\Leftrightarrow\sem s \subsume \sem t\tag{\Cref{thm:concmon}}\\
    &\Rightarrow\forall\phi\in\Jformulas,\, \sem s \satJu \phi \Rightarrow \sem t \satJu \phi.\tag{\Cref{eq:R-closure}}\\
    &\Rightarrow\sem t\satJu\phi(s)\tag{Since $\sem s\satJu\phi(s)$ follows from \Cref{lem:term-to-form}}\\
    &\Leftrightarrow\sem t\revsubsume\sem s \tag{\Cref{lem:term-to-form}}\\
    &\Leftrightarrow \axinf\vdash s \leq t \tag{\Cref{thm:concmon}}
  \end{align*}
  \qedhere
\end{proof}

This extends to sets of pomsets in the following sense:
\begin{restatable}{proposition}{adequacyK}
  Given two terms $e,f\in\terms$, the following equivalences hold:
  \begin{align}
    \axsr\vdash e\leq f&\Leftrightarrow\paren{\forall \phi,\,
                         e\satKE \phi \Rightarrow f \satKE \phi}\Leftrightarrow \paren{\forall \phi,\,
                         f\satKU \phi \Rightarrow e\satKU \phi}\label{eq:adeqKinf}\\
    \axsr\vdash e = f&\Leftrightarrow \paren{\forall \phi,\,
                       e\satKE \phi \Leftrightarrow f\satKE \phi}\Leftrightarrow \paren{\forall \phi,\,
                       e\satKU \phi \Leftrightarrow f\satKU \phi} \label{eq:adeqKeq}\\
    \axsrinf\vdash e\leq f
    &\Leftrightarrow\paren{\forall \phi,\,
      e\satJuE \phi \Rightarrow f \satJuE \phi}
      \Leftrightarrow\paren{\forall \phi,\,
      f\satJdU \phi \Rightarrow e \satJdU \phi}\label{eq:adeqJinf}\\
    \axsrinf\vdash e = f
    &\Leftrightarrow \paren{\forall \phi,\,
      e\satJuE \phi \Leftrightarrow f\satJuE \phi}
      \Leftrightarrow \paren{\forall \phi,\,
      e\satJdU \phi \Leftrightarrow f\satJdU \phi} \label{eq:adeqJeq}
  \end{align}
\end{restatable}
\begin{proof}
  \begin{description}
  \item[\eqref{eq:adeqKinf}] We prove both directions. 
    \begin{description}
    \item[$(\Rightarrow)$] Assume $\axsr\vdash e\leq f$. By
      \Cref{cor:bi-sr} this means $\sem e\subsetsim\sem
      f$. Therefore, we may conclude by \Cref{eq:semK}.
    \item[$(\Leftarrow)$] We show that each LHS implies $\sem e\subsetsim\sem
      f$; that is, $\axsr\vdash e \leq f$:
      \begin{itemize}
      \item Assume
        $\forall \phi,\, f\satKU \phi \Rightarrow e\satKU \phi$. Then
        in particular, since $\sem f\subsetsim \sem f$ by
        \Cref{lem:term-to-form} we have $f\satKU \Phi(f)$, hence
        $e\satKU\Phi(f)$ ergo $\sem e\subsetsim \sem f$; 
      \item Assume
        $\forall \phi,\, e\satKE \phi \Rightarrow f \satKE \phi$, and
        let $P\in\sem e$. By \Cref{lem:expr-to-term} we know
        that there is $s\in T_e$ such that $P\pomeq\sem s$, and by
        \Cref{lem:term-to-form} we get $e\satKE\phi(s)$. Hence
        $f\satKE\phi(s)$. Hence, by~\Cref{eq:phi}, we get
        $P\pomeq\sem s\in\sem f$.
      \end{itemize}

    \end{description}
  \item[\eqref{eq:adeqKeq}] follows from \eqref{eq:adeqKinf},
    and the fact that $\leq$ is antisymmetric.\qedhere
  \end{description}
\end{proof}

\section{Local Reasoning} \label{sec:local}
Some of the discussions in this section do not rely on which satisfaction
relation we pick. When this is the case, we use the symbol $\models$ to mean any of the relations $\satK,\satJu,\satJd$.

\subsection{Modularity}
\label{sec:modularity}

Pomset logic enjoys a high level of compositionality, much like
algebraic logic. Formally, this comes from the following principle:
\[\text{If }e\models\phi\text{ and }\forall a,\,\sigma a\models\tau a,\text{ then }\hat\sigma e\models \hat\tau \phi.\]
This makes possible the following verification scenario: Let $P$ be a
large program, involving a number of simpler sub-programs
$P_1,\dots,P_n$. We may simplify $P$ by replacing the sub-programs by
uninterpreted symbols $x_1,\dots,x_n$. We then check that this
simplified program satisfies a formula $\Phi$, the statement of which
might involve the $x_i$. We then separately determine for each
sub-program $P_i$ some specification $\phi_i$. Finally, using the
principle we just stated, we can show that the full program $P$
satisfies the formula $\Phi'$, obtained by replacing the $x_i$ with
$\phi_i$. 

      %       local reasoning intro. 

\subsection{Frame rule}
\label{sec:frame}

A key of objective of applied, modelling-oriented, work in logic and 
semantics is to understand systems --- such as complex programs, 
large-scale distributed systems, and organizations --- compositionally. 
That is, we seek understand how the system is made of components 
that can be understood independently of one another. A key aspect 
of this is what has become known as \emph{local reasoning}.  That 
is, that the pertinent logical properties of the components of a system 
should be independent of their context. 

In the world of Separation Logic \cite{Rey02,IO01,CSL}, 
for reasoning about how computer programs manipulate memory, O'Hearn, 
Reynolds, and Yang \cite{OHRY2001} suggest that 
\begin{quote}
  `To understand how a program works, it should be possible for reasoning and 
  specification to be confined to the cells that the program actually accesses. 
  The value of any other cell will automatically remain unchanged.'
\end{quote}
In this context, a key idea is that of the `footprint' of a program; that is, that 
part of memory that is, in an appropriate sense, used by the program 
\cite{RazaGardner2008}. If, in an appropriate sense, a program executes 
correctly, or `safely', on its footprint, then the so-called `frame property' 
ensures that the resources present outside of the footprint and, by implication, 
their inherent logical properties, are unchanged by the program.    

In the setting of Separation Logic, the frame property is usually represented by
a Hoare-triple rule of the form 
\[
  \frac{ \{ \phi \} C \{ \psi \} }{ \{ \phi * \chi \} C \{ \psi * \chi \} } 
  \quad \mbox{\rm $C$ is independent of $\chi$.} 
\]
That is, the formula $\chi$ does not include any variables (from the memory) 
that are modified by the program $C$. 

In order to formulate the frame property in our framework, we first
fix the notion of independence between a program and a formula. We say
that a pomset $P$ is \emph{$R$-independent} of a formula $\phi$,
written $P\independent[R] \phi$ if $P\not\models_R \context{\boxing\phi}$. 
Since independence is meant to prevent overlap, the use of the $\context{-}$ 
modality should come as no surprise.

To explain the need for the $\boxing{-}$ modality, first consider a pomset 
$P$ satisfying $\context\phi$. To extract a witness of this fact, we must 
remove parts of $P$, including boxes and events, such that the remainder 
satisfies $\phi$. However, there are no restrictions on the relationship 
between the remaining events and those we have deleted. In a sequence 
of three events, we are allowed to keep the two extremities, and delete 
the middle one. In contrast, to get a witness of $\context{\boxing\phi}$, 
we need to identify a box on $P$ whose contents satisfy $\phi$, and 
remove all events external to that box. The result is that the deleted 
events, that is, the context of our witness, can only appear outside the 
box, and must treat all events inside uniformly. In other words, these 
events can interact with the behaviour encapsulated in the box, but 
cannot interact with individual components inside. For this reason, the 
frame properties given in Proposition~\ref{lem:framerule} are expressed 
using $\boxing \phi$  --- that is, the encapsulation of $\phi$ --- rather 
than $\phi$.

With this definition, we can now state three frame rules, enabling
local reasoning with respect to the parallel product, sequential
prefixing, and sequential suffixing.
\begin{restatable}[Frame properties]{proposition}{framerule}\label{lem:framerule}
  If $P\independent\phi$, and $Q\satK\boxing \phi$, then it holds that:
  \begin{enumerate}[(i)]
  \item $\forall \psi\in\Kformulas,\,P\satK\psi\Leftrightarrow P\pompar Q\satK \psi\nextto\boxing\phi$;
  \item $\forall \psi\in\Kformulas,\,P\satK\psi\Leftrightarrow P\pomseq Q\satK \psi\then\boxing\phi$;
  \item $\forall \psi\in\Kformulas,\,P\satK\psi\Leftrightarrow Q\pomseq P\satK \boxing\phi\then\psi$.
  \end{enumerate}
  % If $P\independent\phi$, and $Q\satK\boxing \phi$, then for any $\psi\in\Kformulas$ t.f.a.e.:
  % \begin{enumerate}[(i)]
  % \item $P\satK\psi$;
  % \item $P\pompar Q\satK \psi\nextto\boxing\phi$;
  % \item $P\pomseq Q\satK \psi\then\boxing\phi$;
  % \item $Q\pomseq P\satK \boxing\phi\then\psi$.
  % \end{enumerate}
\end{restatable}
\begin{proof}
  Clearly, if $P\satK\psi$, since $Q\satK\boxing \phi$ we have
  immediately that $P\pompar Q\satK \psi\nextto\boxing\phi$,
  $P\pomseq Q\satK \psi\then\boxing\phi$, and
  $P\pomseq Q\satK \boxing\phi\then\psi$.

  Now assume $P\pompar Q\satK \psi\nextto\boxing\phi$. This means that
  there exists $P',Q'$ such that $P\pompar Q \pomeq P'\pompar Q'$,
  $P'\satK\psi$ and $Q'\satK\boxing \phi$. The factorization
  $P\pompar Q \pomeq P'\pompar Q'$ may be further decomposed into
  pomsets $P_1,P_2,Q_1$, and $Q_2$ such that $P\pomeq P_1\pompar P_2$,
  $Q\pomeq Q_1\pompar Q_2$, $P_1\pompar Q_1\pomeq P'$, and
  $P_2\pompar Q_2\pomeq Q'$.

  Note that $\phi$ cannot be satisfied by the empty pomset: otherwise,
  since the empty pomset is contained in any other pomset, we would
  have $P\satK\context{\boxing\phi}$. Therefore since both $Q'$ and
  $Q$ satisfy $\boxing\phi$, both have to be non-empty boxes. As a
  result, we get:
  \begin{itemize}
  \item since $Q\pomeq Q_1\pompar Q_2$, either $Q_1$ or $Q_2$ have to be empty;
  \item since $Q'\pomeq P_2\pompar Q_2$, either $P_2$ or $Q_2$ have to be empty.
  \end{itemize}
  If $Q_2$ is empty, we violate again the hypothesis that $P$ is
  independent from $\phi$, since
  $P\pomeq P_1\pompar P_2\pomeq P_1\pompar Q'\contains Q'$, and
  $Q'\satK\boxing \phi$.

  Hence we know that $Q_1$ and $P_2$ are empty, meaning that
  $P\pomeq P_1\pomeq P'\satK\psi$.

  The same argument works for the sequential product.
\end{proof}
\begin{remark}
  Note that this lemma does not hold for $\subsume$ or $\revsubsume$
  instead of $\simeq$. The left-to-right implications always hold, but
  the converse may be fail, as we now demonstrate with some examples.
  \begin{itemize}
  \item For $\subsume$, consider the following:
    \begin{mathpar}
      P\eqdef \atom a\and Q\eqdef \boxing{\atom b\pompar\boxing{\atom c}}\and \phi\eqdef\context c \and \psi\eqdef a \nextto b.
    \end{mathpar}
    We may check that in this case we have:
    \begin{mathpar}
      P\not\satJd \context{\boxing\phi}\and Q\satJd\boxing\phi\and
      P\pompar Q\satJd\psi\nextto\boxing\phi\and P\not\satJd\psi.
    \end{mathpar}
    This happens because in order to satisfy $\psi\nextto\boxing\phi$
    we may rearrange the pomsets using the ordering. More precisely, we
    use the fact that 
    \[P\pompar Q= \atom a\pompar \boxing{\atom b\pompar\boxing{\atom
          c}}\subsume \paren{\atom a\pompar\atom
        b}\pompar\boxing{\atom c}.\]
  \item For $\revsubsume$, we have a similar example. The difference is
    that instead of $Q$ overspiling into $P$, we have the converse.
    \begin{mathpar}
      P\eqdef \atom a \pompar\atom b\and
      Q\eqdef \atom c\and
      \phi\eqdef \context c\and
      \psi\eqdef \atom a.
    \end{mathpar}
    Since $P$ does not contain any $c$, but has more than one event,
    $P\not\satJd \context{\boxing\phi}\vee\psi$. On the other hand
    $Q\satJd\boxing\phi$ because
    $Q\revsubsume\boxing Q$ and $Q\contains Q\satJu c$.
    Finally,
    $P\pompar Q=\atom a \pompar\atom b\pompar \atom c\revsubsume \atom
    a \pompar\boxing{\atom b\pompar \atom c}$, and we have both
    $\atom a \satJd\psi$ and
    $\boxing{\atom b\pompar\atom c}\satJd\boxing\phi$.
  \end{itemize}
\end{remark}
\begin{remark}
  this principle may be extended to sets of pomsets. Indeed, if we
  define the independence relation for sets of pomsets as
  $A\independent[R]\phi\eqdef \forall P\in A,\,P\independent[R]\phi$,
  then Proposition \ref{lem:framerule} holds for both $\satKU$ and
  $\satKE$.
\end{remark}

                %                 The modality $\context{\_}$ provides an explicit way of performing
                %                 local reasoning (in the sense of (Concurrent) Separation Logic 
                %                 \cite{IO01,Rey02,CSL,BrookesO'Hearn2016}). 
                %                 Indeed, if a program satisfies some formula $\phi$, then, if we insert 
                %                 this program in any context, the resulting program will satisfy 
                %                 $\context\phi$. (Relationships 
                %                 between Concurrent Separation Logic and CKA have bee explored in 
                %                 \cite{OPVH2015}.) 

                %                 To illustrate how the box modality relates to protection, consider the
                %                 classical formula $\context{\boxing{\neg\context{a}}}$. A poset
                %                 satisfies this formula iff it contains a box that do not contain any
                %                 $a$-labelled event. Note that the negation is used here to express a
                %                 closed-world property, which is not intrinsically classical, but
                %                 cannot be expressed in our intuitionistic language. If a program $p$
                %                 satisfies $\context{\boxing{\neg\context{a}}}$, then any program
                %                 featuring $p$ positively also satisfies it. This can be understood as
                %                 saying that if $p$ contains a correct box, then no use of $p$ can
                %                 introduce an $a$-bug inside that box, justifying our claim that boxes
                %                 embody protection.

\subsection{Example}

In this section, we present an example program, and showcase reasoning
principles of pomset logic. In particular, we will highlight the use of
local reasoning when appropriate.

Consider the following voting protocol: a fixed number of voters,
$v_1,\dots,v_n$, are each asked to increment one of the counters
$c_1,\dots,c_k$. The tally is then sent to each of the $v_i$, to
inform them of the result. The increment is implemented similarly to
our running example of the distributed
counter~(Example~\ref{ex:distrib-count}).
The implementation of the protocol is displayed in Figure~\ref{fig:voting}, 
together with the intended semantics of the atomic actions. 

\begin{figure}[t]
  \centering
  \fbox{
    \begin{minipage}{.8\linewidth}
      \begin{align*}
        \VoteProc&\eqdef~ \Choose \tseq \Publish\\
        \Choose&\eqdef~ \boxing{\vote 1}\tpar\cdots\tpar\boxing{\vote n}\\
        \vote i&\eqdef~ \sum_{1\leqslant j\leqslant k}\choose{i,j}\tseq \Read{j}\tseq\Compute{}\tseq\Write{j}\\
        \Publish&\eqdef~ \send 1\tpar\cdots\tpar\send n
      \end{align*}
      \begin{tabular}{c@{:~}l}
        $\send i$& the contents of counters $c_1,\dots,c_k$ is sent to voter $v_i$\\
        $\choose{i,j}$& voter $v_i$ chooses counter $c_j$\\
        $\Read{j}$& the content of counter $c_j$ is loaded into a local variable\\
        $\Compute{}$& the local variable is incremented\\
        $\Write{j}$& the content of the local variable is stored in counter $c_j$
      \end{tabular}
    \end{minipage}
  }
  \caption{Voting protocol}
  \label{fig:voting}
\end{figure}

\paragraph*{Conflict}
As in Example~\ref{ex:distrib-count}, if we forgo the boxes in $\Choose$, we cannot enforce 
mutual exclusion. Recall that the undesirable behaviour is captured by the following formula:
\[
	\conflict{j}\eqdef\context{\paren{\Read {j}\nextto\Read {j}}\then\paren{\Write {j}\nextto \Write {j}}} 
\]
We may see this by defining an alternative (faulty) protocol:
\[
	\VoteProc'\eqdef \paren{{\vote 1}\tpar\cdots\tpar{\vote n}}\tseq\Publish
\]
and then checking that this protocol displays the behaviour we wanted to avoid:
\[
	\VoteProc'\satJuE \conflict{j}.
\]
This statement should be read as `there is a pomset in $\sem {\VoteProc'}$ that
is larger than one containing a conflict'.
We can show the existence of this `bug' by local reasoning. We may
first prove that $\vote i\tpar\vote {i'}\satJuE \conflict{j}$ (for some arbitrary $i\neq i'$). The
properties of $\context{-}$ then allow us to deduce that
\[\VoteProc' \pomeq \paren{\vote i\tpar\vote
    {i'}\tpar\cdots}\tseq\cdots\satJuE \context{\conflict{j}}\equiv
  \conflict{j}.\]
The implementation in Figure~\ref{fig:voting} avoids this problem, and
indeed it holds that:
\[\VoteProc\not\satJuE \conflict{j}.\]
However, showing that this formula is \emph{not} satisfied by the
program is less straightforward and, in particular, cannot be done
locally: we have to enumerate all possible sub-pomsets, and check that
none provide a suitable witness.

\paragraph*{Sequential separation}
In our protocol, the results of the vote are only communicated after
every participant has voted. This is specified by the following
statement:
\[\seqsep\eqdef\context{\neg\paren{\bigvee_i \send{i}}}\then\context{\neg\paren{\bigvee_{i,j}{\choose{i,j}}}}.\]
This may be checked modularly. Indeed, one may prove by simple
syntactic analysis that
\begin{mathpar}
  \Choose\satKU\context{\neg\paren{\bigvee_i \send{i}}}\and\text{and}\and
  \Publish\satKU\context{\neg\paren{\bigvee_{i,j}{\choose{i,j}}}}.
\end{mathpar}
Therefore, we may combine these to get that:
\[\VoteProc=\Choose\tseq\Publish\satKU\context{\neg\paren{\bigvee_i \send{i}}}\then\context{\neg\paren{\bigvee_{i,j}{\choose{i,j}}}}=\seqsep.\]

For voter $i$, two of the most meaningful steps are $\choose{i,j}$ and
$\send{i}$, i.e. when the vote is cast and when the result of the vote
is forwarded to them. Using the macro
$\fchoose i\eqdef\bigvee_j\choose{i,j}$, we can specify that during
the protocol, each voter first votes, and then gets send the result:
\[\votethensend\eqdef\context{\paren{\fchoose 1\then\send 1}\nextto\dots\nextto\paren{\fchoose n\then\send n}}.\]

\paragraph*{Unique votes}

Another important feature of this protocol is that each voter may only
cast a single vote. Knowing that each voter controls a single box, we
express this property with the statement:
\[
	\VoteProc\not\satJdE\bigvee_{j,j'}\context{\boxing{\context{\Write{j}\nextto\Write{j'}}}}.
\]
Since we use the relation $\satJdE$ with the connective $\nextto$, we
allow any possible ordering of the two write events. The only 
constraint is that there should be at least two of them in the same
box. As for the `conflict' property, if the `bad' behaviour were to
happen, one could prove it compositionally. However, disproving the
existence of such a behaviour is a more global process, involving the
exploration of all possible sub-pomsets.

\paragraph*{Frame property}

As we have seen in previous examples, proving that a formula does not
hold can be challenging, because the non-existence of a local pattern
is \emph{not} a local property. We may circumvent this problem by
adding more boxes in both programs and formulas.  This is related to a
common pattern in parallel programming: in a multi-threaded program,
one may insert fences to `tame' concurrency. Doing so simplifies
program analysis, at the cost of some efficiency. Similarly, since
adding boxes restricts behaviours --- thus disallowing some possible
optimizations --- the analysis of a program becomes simpler and more
efficient.

We illustrate this with the following statement:
\begin{mathpar}
	\boxing\Choose\tseq\boxing\Publish\not\satKU\context{\Write{}\then\Write{}}\then\boxing\phi\and\text{where }	
	\phi\eqdef\neg\paren{\emptyprop\vee\context{\boxing{\context{\Write{}}}}}.
\end{mathpar}

\noindent $\context{\Write{}\then\Write{}}$ indicates that two `write' instructions can be
executed in sequence, while $\phi$ denotes a non-empty pomset, not containing any boxes
with a `write' event inside. We can first prove properties of the subprograms:
\begin{mathpar}
  \boxing{\Publish}\satKU\boxing{\phi}\and
  \boxing{\Choose}\not\satKU\context{\boxing{\phi}}\and
  \boxing{\Choose}\not\satKU\context{\Write{}\then\Write{}}.
\end{mathpar}
Since $\boxing\Choose\independent\phi$ and
$\boxing\Publish\satKU\boxing\phi$, we obtain from the frame rule that
\[
\boxing\Choose\tseq\boxing\Publish\satKU\context{\Write{}\then\Write{}}\then\boxing\phi\Leftrightarrow\boxing{\Choose}\satKU\context{\Write{}\then\Write{}}.
\]
Since we have locally disproved the latter, we may deduce that the former does not hold.

\section{Future work}
\label{sec:future}

In this paper, we have not considered the CKA operator $-^\star$. A
natural further step would be to do so, with the corresponding need to
consider versions of pomset logic with fixed points.
Connections with Hoare-style program logics, such as Concurrent
Separation Logic~\cite{BrookesO'Hearn2016,OPVH2015} with its concrete
semantics, should also be considered.

Our satisfaction relation over pomsets is defined inductively. However, the 
satisfaction relations we define for sets of pomsets is not: we define in terms 
of the former relation. For practical purposes, such as model-checking, it 
would be useful to have a similar inductive definition for sets of pomsets.

It is also worth noticing that the definitions and statements in
Section~\ref{sec:algebra} are straight-forward generalizations of
their counterparts in CKA (without boxes); even the proofs of those
results follow a similar strategy. However, we could reuse almost no
result from CKA: instead we had to reprove everything from
scratch. This situation is deeply unsatisfactory, and we plan on
investigating techniques to better `recycle' proofs in this
context. Recent work on (C)KA with
\emph{hypotheses}~\cite{dkpp19,kbswz20} seems to be a step towards 
this goal.

\bibliography{bibli}

\clearpage
\appendix
\section{Proof of Lemma~\ref{lem:term-to-form}}
\allowdisplaybreaks
\termtoform*
\begin{proof}
  By induction on $s$:
  \begin{align*}
    s&=1:
    &P\satK\phi(1)=\emptyprop
    &\Leftrightarrow P\pomeq\unitposet=\sem 1.\\
     &&P\satJu\phi(1)=\emptyprop
    &\Leftrightarrow P\pomeq\unitposet\Leftrightarrow P\revsubsume\unitposet=\sem 1.\\
     &&P\satJd\phi(1)=\emptyprop
    &\Leftrightarrow P\pomeq\unitposet\Leftrightarrow P\subsume\unitposet=\sem 1.\\
    s&=a:
    &P\satK\phi(a)=a
    &\Leftrightarrow P\pomeq\atom a=\sem a.\\
     &&P\satJu\phi(a)=a
    &\Leftrightarrow P\revsubsume\atom a=\sem a.\\
     &&P\satJd\phi(a)=a
    &\Leftrightarrow P\subsume\atom a=\sem a.\\
    s&=\boxing t:
    &P\satK\boxing{\phi(t)}
    &\Leftrightarrow P\pomeq\boxing Q
      \wedge Q\satK\phi(t)\\
     &&&\Leftrightarrow P\pomeq\boxing Q
         \wedge Q\pomeq\sem t\\
     &&&\Leftrightarrow P\pomeq \boxing{\sem t}=\sem{\boxing t}.\\
     &&P\satJu\boxing{\phi(t)}
    &\Leftrightarrow P\revsubsume\boxing Q
      \wedge Q\satJu\phi(t)\\
     &&&\Leftrightarrow P\revsubsume\boxing Q
         \wedge Q\revsubsume\sem t\\
     &&&\Leftrightarrow P\revsubsume \boxing{\sem t}=\sem{\boxing t}.\\
     &&P\satJd\boxing{\phi(t)}
    &\Leftrightarrow P\subsume\boxing Q
      \wedge Q\satJd\phi(t)\\
     &&&\Leftrightarrow P\subsume\boxing Q
         \wedge Q\subsume\sem t\\
     &&&\Leftrightarrow P\subsume \boxing{\sem t}=\sem{\boxing t}.\\
    s&=s_1\tseq s_2:
    &P\satK\phi(s_1)\then\phi(s_2)
    &\Leftrightarrow P\pomeq P_1\pomseq P_2\\
     &&&{\color{white}\Leftrightarrow}\;\wedge P_1\satK\phi(s_1)
         \wedge P_2\satK\phi(s_2)\\
     &&&\Leftrightarrow P\pomeq P_1\pomseq P_2
         \wedge P_1\pomeq\sem{s_1}
         \wedge P_2\pomeq\sem{s_2}\\
     &&&\Leftrightarrow P\pomeq \sem{s_1}\pomseq\sem{s_2}=\sem {s_1\tseq s_2}.\\
     &&P\satJu\phi(s_1)\then\phi(s_2)
    &\Leftrightarrow P\revsubsume P_1\pomseq P_2\\
     &&&{\color{white}\Leftrightarrow}\;\wedge P_1\satJu\phi(s_1)
         \wedge P_2\satJu\phi(s_2)\\
     &&&\Leftrightarrow P\revsubsume P_1\pomseq P_2
         \wedge P_1\revsubsume\sem{s_1}
         \wedge P_2\revsubsume\sem{s_2}\\
     &&&\Leftrightarrow P\revsubsume \sem{s_1}\pomseq\sem{s_2}=\sem {s_1\tseq s_2}.\\
     &&P\satJd\phi(s_1)\then\phi(s_2)
    &\Leftrightarrow P\subsume P_1\pomseq P_2\\
     &&&{\color{white}\Leftrightarrow}\;\wedge P_1\satJd\phi(s_1)
         \wedge P_2\satJd\phi(s_2)\\
     &&&\Leftrightarrow P\subsume P_1\pomseq P_2
         \wedge P_1\subsume\sem{s_1}
         \wedge P_2\subsume\sem{s_2}\\
     &&&\Leftrightarrow P\subsume \sem{s_1}\pomseq\sem{s_2}=\sem {s_1\tseq s_2}.\\
    s&=s_1\tpar s_2:
    &P\satK\phi(s_1)\nextto\phi(s_2)
    &\Leftrightarrow P\pomeq P_1\pompar P_2\\
     &&&{\color{white}\Leftrightarrow}\;
         \wedge P_1\satK\phi(s_1)
         \wedge P_2\satK\phi(s_2)\\
     &&&\Leftrightarrow P\pomeq P_1\pompar P_2
         \wedge P_1\pomeq\sem{s_1}
         \wedge P_2\pomeq\sem{s_2}\\
     &&&\Leftrightarrow P\pomeq \sem{s_1}\pompar\sem{s_2}=\sem{s_1\tpar s_2}.\\
     &&P\satJu\phi(s_1)\nextto\phi(s_2)
    &\Leftrightarrow P\revsubsume P_1\pompar P_2\\
     &&&{\color{white}\Leftrightarrow}\;
         \wedge P_1\satJu\phi(s_1)
         \wedge P_2\satJu\phi(s_2)\\
     &&&\Leftrightarrow P\revsubsume P_1\pompar P_2
         \wedge P_1\revsubsume\sem{s_1}
         \wedge P_2\revsubsume\sem{s_2}\\
     &&&\Leftrightarrow P\revsubsume \sem{s_1}\pompar\sem{s_2}=\sem{s_1\tpar s_2}.\\
     &&P\satJd\phi(s_1)\nextto\phi(s_2)
    &\Leftrightarrow P\subsume P_1\pompar P_2\\
     &&&{\color{white}\Leftrightarrow}\;
         \wedge P_1\satJd\phi(s_1)
         \wedge P_2\satJd\phi(s_2)\\
     &&&\Leftrightarrow P\subsume P_1\pompar P_2
         \wedge P_1\subsume\sem{s_1}
         \wedge P_2\subsume\sem{s_2}\\
     &&&\Leftrightarrow P\subsume \sem{s_1}\pompar\sem{s_2}=\sem{s_1\tpar s_2}.
  \end{align*}
  
  Recall that since $\upclpom\_$ and $\clpom\_$ are Kuratowski closure
  operators, they distribute over unions. We may thus obtain:
  \begin{align*}
    \psem[\K]{\Phi(e)}&=\bigcup_{s\in T_e}\psem[\K]{\phi(s)}\pomeq\bigcup_{s\in T_e}\set{\sem s}\pomeq\sem e.\\
    % \psem[\Ju]{\Phi(e)}&=\bigcup_{s\in T_e}\psem[\Ju]{\phi(s)}\pomeq\bigcup_{s\in T_e}\upclpom{\set{\sem s}}\pomeq\upclpom{\sem e}.\\
    \psem[\Jd]{\Phi(e)}&=\bigcup_{s\in T_e}\psem[\Jd]{\phi(s)}\pomeq\bigcup_{s\in T_e}\clpom{\set{\sem s}}\pomeq\clpom{\sem e}.
  \end{align*}
  The statements then follow by~\eqref{eq:satKU} and~\eqref{eq:satJdU}.\qedhere
\end{proof}

\end{document}